\documentclass[preprint,11pt,authoryear]{elsarticle}
\usepackage{amsmath}
\usepackage{amsthm}
\allowdisplaybreaks
\usepackage{minted}
\usepackage{orcidlink}
\usepackage[utf8]{inputenc}

\usepackage{amssymb}
\usepackage{url}
\usepackage{adjustbox}
\usepackage{amsthm}
\newtheorem{proposition}{Proposition}
\usepackage{amssymb}
\usepackage{graphicx}
\usepackage{bm}
\usepackage{rotating}
\usepackage{url}
\usepackage{amssymb}
\usepackage{hyperref}
\usepackage{lscape}
\usepackage{hyperref}
\usepackage{multirow}
\usepackage{tfrupee}
\usepackage{subcaption}
\usepackage{amsthm}
\usepackage{comment}
\usepackage[english]{babel}
\usepackage{setspace}
\usepackage{enumerate}
\usepackage{epstopdf}
\usepackage{color}
\usepackage{scalerel,stackengine}
\usepackage{mathtools}
\usepackage{url}
\usepackage[normalem]{ulem}
\usepackage{subcaption}
\usepackage{verbatim}
\usepackage{float}
\usepackage{esvect}

\usepackage[utf8]{inputenc}
\usepackage{amsmath}

\usepackage{url}
\usepackage{mathrsfs}
\usepackage{booktabs}
\usepackage{longtable}
\usepackage{hyperref}
\usepackage[labelfont=bf, textfont=normal]{caption}
\hypersetup{colorlinks,linkcolor={blue},citecolor={blue},urlcolor={red}}  
\usepackage{amsmath}
\usepackage{algorithm}
\usepackage{color}
\usepackage{algpseudocode}
\usepackage{amsthm,setspace}
\newtheorem{theorem}{Theorem}

\newtheorem{remark}{Remark}
 \usepackage{lscape}
 
\usepackage[margin=2.22cm]{geometry}

\usepackage{lineno}
\makeatletter
\def\ps@pprintTitle{%
 \let\@oddhead\@empty
 \let\@evenhead\@empty
 \def\@oddfoot{\reset@font\hfil\thepage\hfil}
 \let\@evenfoot\@oddfoot
}
\makeatletter

\begin{document}

\begin{frontmatter}
\title{\textcolor{black}{Robust Modeling of Extremes in the Presence of Inliers with Enhanced Tail Estimation}\\
}
\author[label1]{Shivshankar Nila\,\orcidlink{0009-0004-0465-5490}}
\author[label1]{Ishapathik Das\corref{cor1}}
\ead{ishapathik@iittp.ac.in}
\cortext[cor1]{Corresponding author: Ishapathik Das}
\author[label1]{N. Balakrishna}

\address[label1]{Department of Mathematics and Statistics, Indian Institute of Technology Tirupati, Tirupati-517619, Andhra Pradesh, India.}
\begin{abstract}
Extreme value theory provides a fundamental framework for modeling rare and extreme events; however, threshold selection remains a persistent challenge, particularly in the presence of inliers such as instantaneous or early failures. Such observations commonly arise in applications including reliability studies and environmental data, where clusters of observations near the origin or at the origin can substantially distort classical threshold selection procedures and tail inference. In this paper, we propose a robust modeling framework that accounts for inliers, extremes, and the tail proportion. Parameter estimation is carried out using maximum likelihood. The proposed methodology is compared with classical numerical and graphical diagnostic tools, including the mean excess plot, parameter stability plot, Hill plot, and Pickands plot, as well as existing extreme value mixture models.
The theoretical properties of the proposed model are established, and its performance is evaluated through extensive Monte Carlo simulations and real-data applications.  The results demonstrate that the proposed methodology provides more accurate threshold estimation and more reliable extreme-value inference in the presence of inliers compared with existing classical approaches. Overall, the proposed methodology provides more accurate threshold estimation and tail inference in the presence of inliers, addressing key limitations of existing methods.

\end{abstract}

\begin{keyword}
Extreme value theory, Threshold estimation, Graphical diagnostics, Mixture models, Inliers, Tail inference, Risk measures.
\end{keyword}
\end{frontmatter}
\section{Introduction}\label{sec1}
\noindent
Extreme value theory (EVT) provides a statistical framework for estimating the occurrence and magnitude of rare or extreme events that lie beyond the range of available observations, and it aids in long-term prediction. EVT provides a well-established theoretical foundation on which we can build statistical models describing extreme events and offers a scientific alternative to pure guesswork; for details, see \citet{coles2001introduction}.
 Traditional statistical theory is primarily concerned with the centre of probability distributions, and the tail regions are often overlooked. In contrast, EVT focuses on the statistical properties of events with exceptionally high or low magnitudes that occur with low frequency.  The literature on EVT has grown substantially over the last few years, driven by its wide range of applications in hydrology, meteorology, environmental studies, actuarial science, economics, finance, traffic prediction, insurance, and structural engineering. Classical extreme value models are asymptotically motivated and specifically developed to characterise the distributional behaviour of extreme events(\citet{hmingthansanga2026gamma}).

Even though EVT is based on asymptotic modeling to quantify extreme events, selecting an appropriate threshold in the peak-over-threshold (POT) method remains a key challenge for accurately approximating the population tail (\citet{murphy2025automated,schneider2021threshold}.
If the selected threshold is low, it leads to poor GPD approximation and biased return-level estimates, whereas a high threshold increases parameter variance due to fewer exceedances.
Several approaches for threshold selection have been explored, including numerical methods, graphical techniques, and mixture models.
In numerical methods for threshold selection, the $90$th percentile,  one may select the $k$ largest observations, with $k = \sqrt{n}$ (\citet{ferreira2003optimising}) or $k = \frac{\sqrt[2/3]{n}}{\log(\log(n))}$ (\citet{benito2023assessing}).
These methods are simple but often unreliable because datasets differ in their characteristics. Several classical graphical diagnostics, such as the mean excess plot (MEP) (\citet{coles2001introduction}), Pickands plot (PP) (\citet{finkenstadt2003extreme}), Hill plot (HP) (\citet{finkenstadt2003extreme}), and parameter stability plot (PSP) (\citet{coles2001introduction}), are commonly used for threshold selection.
However, these approaches are highly subjective and require considerable expertise, thereby typically ignoring uncertainty in threshold selection.  
Firstly, \citet{behrens2004bayesian} proposed the \textit{extreme value mixture model (EVMM)}, which treats the threshold as an additional parameter and accounts for the uncertainty associated with its estimation.
EVMM consists of two components, a flexible parametric model for observations below the threshold (the ``bulk model'') and a GPD for exceedances above the threshold (the ``tail model''). Various EVMMs with different bulk distributions have been proposed in the literature; some are \citet{carreau2009hybrid,tancredi2006accounting,cabras2011bayesian,solari2012unified,lee2012modeling}, and \citet{do2012semiparametric}.
 More flexible approaches have also been developed using kernel density estimators and Dirichlet process mixtures, as described in \citet{macdonald2012extreme,fuquene2015semi}, and \citet{hu2018evmix}.
While these studies retain the GPD for tail modeling, recent work has studied various bulk and tail components (\citet{marambakuyana2024composite}).
However, the robustness of the tail fit relative to the bulk distribution remains a major concern.  Detailed reviews of threshold selection and mixture models are given in \citet{scarrott2012review} and \citet{dey2016extreme}. 
Accurate estimation of extreme quantiles is crucial when the distribution's tail is of interest. Several studies have improved tail estimation using alternative distributions and estimation procedures, see \cite{pinheiro2016comparative} and \cite{jokiel2024estimation}.
Recently, several studies have focused on tail index and threshold estimation approaches, some of which are presented in \citet{northrop2014improved,schneider2021threshold} and \citet{wu2025heavy}.
More recently, another approach addresses threshold selection but ignores non-exceedances and associated uncertainty, potentially leading to suboptimal inference(\citet{sakthivel2025dual,tukey1977exploratory,andrews2015robust}). Recent studies have explored flexible models related to heavy-tailed data (\citet{hmingthansanga2026gamma}).

Many random events occur in real-life scenarios, such as in life-testing experiments, where an item might fail immediately, resulting in a zero or near-zero lifetime. In rainfall data, days with little to no rain can occur during the rainy season. These dry days or fewer rainy days may cause serious problems, such as drought. 
  Real-world datasets frequently exhibit a high concentration of observations near a particular point \( x_{0} \), whereas the remaining observations follow varying distributions. We refer to such observations as  "early failures" or  "inliers near the point \( x_{0} \)".
  When \( x_{0} = 0 \) and observations at particular \( x_{0} = 0 \), this is specifically termed an "instantaneous failure" or an  "inliers at zero." 
 In such situations, inliers often emerge as a subset of observations that appear inconsistent with the rest of the dataset.
  To accommodate such inlier scenarios for improved inference, \citet{muralidharan2006analysis} proposed a modified Weibull framework.

\subsection*{Study motivation}
\noindent There is further scope for study, as there is no generalised framework for defining EVMMs and the tail fraction above the threshold, both of which are crucial for estimating extreme quantiles. Distributional form misspecification below the threshold affects the tail fraction and the estimation of extreme-value parameters. The robustness of both bulk and tail fits remains a major concern \citep{hu2013extreme, hu2018evmix}. Existing EVMMs rarely account for inliers within the bulk component, which may affect threshold estimation and extreme-value inference. Robust models capable of fitting across different situations are therefore recommended by \cite{behrens2004bayesian}, and \cite{hu2018evmix}. \citet{naveau2016modeling} models low, moderate, and heavy positive rainfall amounts without threshold selection, explicitly excluding dry days; the authors note that incorporating rainfall occurrence (wet/dry events) would be an important extension.

To illustrate the impact of inliers on threshold selection and extreme-value parameter estimation, we consider a motivating dataset containing inliers and characterised by a known threshold \(u = 11.49\), scale parameter \(\sigma = 5\), and shape parameter \(\xi = 0.20\) of the GPD. Figure~\ref{fig:meanexcess} presents the diagnostic plots corresponding to the MEP, PSP, PP, and HP. The figure demonstrates that these classical graphical diagnostic methods may yield misleading estimates of thresholds and tail parameters in the presence of inliers. This highlights the need for a more robust extreme-value modeling framework when the underlying data-generating process contains inliers.

\begin{figure}[!t] 
    \centering
    \includegraphics[width=1.0\textwidth,height=10cm]{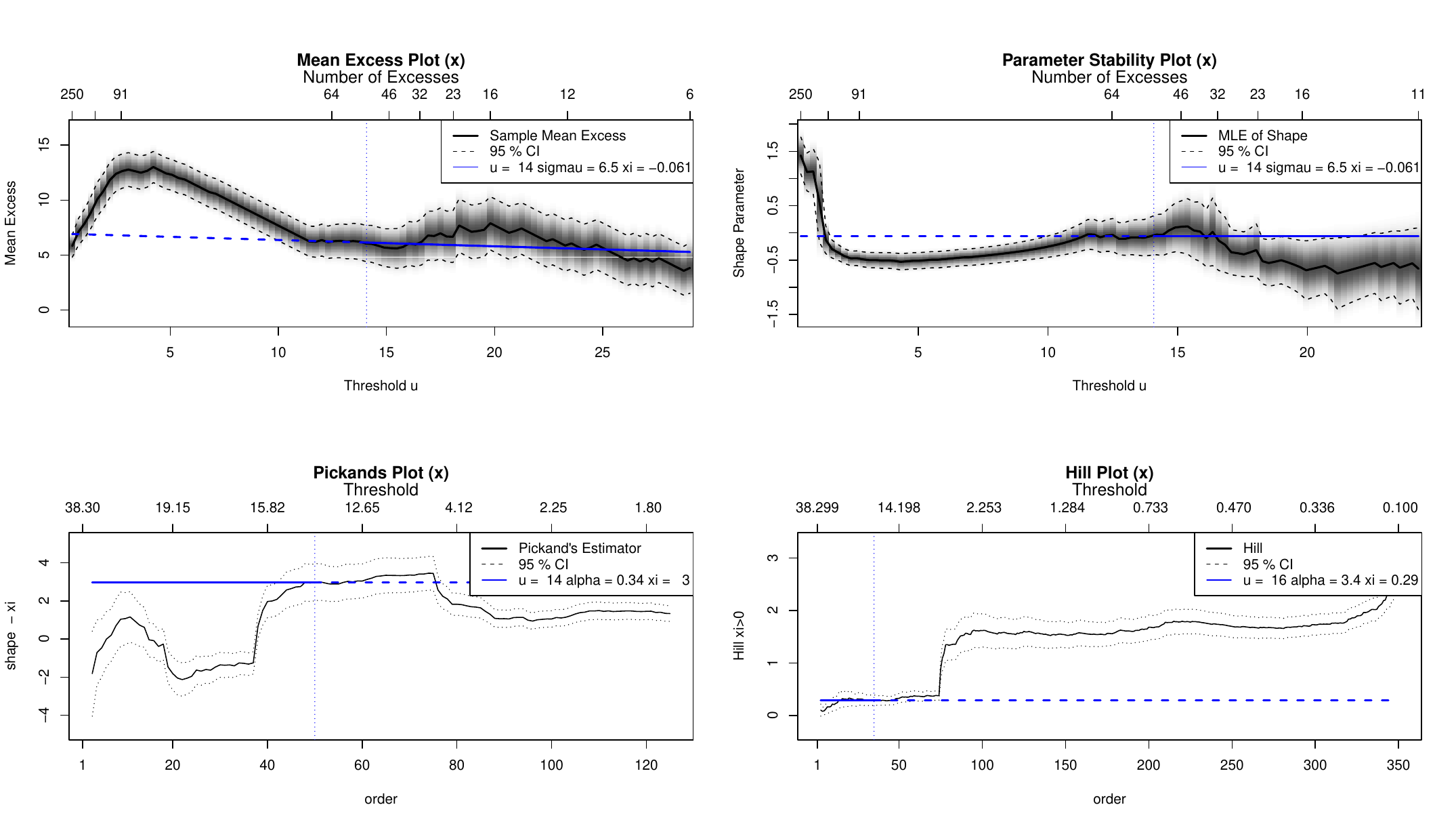}
\caption{
Diagnostic plots of MEP, PSP, PP, and HP for a known dataset characterised by a true threshold \(u = 11.49\), scale parameter \(\sigma = 5\), and shape parameter \(\xi = 0.20\), in the presence of inliers. The results show that the classical graphical methods may provide misleading estimates of the extreme-value parameters \((\xi,\sigma)\).
}
    \label{fig:meanexcess}
\end{figure}


Recently, \citet{nila2025modeling} proposed an extreme value inlier mixture model (EVIMM) for modeling extremes in the presence of instantaneous failures. Their study demonstrated that incorporating inliers improves threshold and tail estimation compared to existing EVMMs. However, when both extreme and non-extreme observations are of interest, accurate modeling of the full distribution becomes essential, as highlighted by \citet{andre2024joint}. Motivated by the need for full distribution modeling discussed in \citet{andre2024joint}, the open problem identified in \citet{nila2025modeling}, and the rainfall scenarios considered in \citet{naveau2016modeling}, we propose a novel extended extreme value inlier mixture model (EEVIMM) that captures inliers, moderate observations, and extreme events within a unified framework. Additionally, we evaluate robustness, goodness-of-fit (GoF), and parameter-estimation accuracy by comparing the proposed EEVIMM with established approaches, including the EVMM \citep{behrens2004bayesian}, the flexible EVMM (FEVMM) \citep{macdonald2011flexible}, the EVIMM \citep{nila2025modeling}, the model proposed by \citet{do2012semiparametric}, and the classical graphical methods. Our analysis examines whether the proposed model improves threshold and tail estimation. Performance is further evaluated using standard GoF tests, including the Anderson-Darling ($AD$), Cramér-von Mises ($CvM$), and Kolmogorov-Smirnov ($KS$) tests, as well as information criteria.

The remainder of this paper is structured as follows. Section~\ref{background} presents background on classical extreme value models and POT. Section~\ref{methodology} presents the proposed methodology, including the model construction, structural properties, risk measures such as Value-at-Risk and Expected Shortfall, algorithms for data simulation and parameter estimation, and the asymptotic properties of the maximum likelihood estimator (MLE). Section \ref{SStudy} presents an extensive Monte Carlo simulation study to evaluate the performance of the proposed methodology. Section \ref{Applications} presents an application of the proposed methodology to real data. Finally, Section~\ref{conclusion} provides a concluding remark and possible future directions.
\section{Background} \label{background}
\subsection{ The classical extreme value model}\label{evt}\label{background1}
\noindent
The popular approach for modeling extremes in EVT is the POT method, which models exceedances above a threshold to analyse tail behaviour, and this methodology is supported by a key result of \citet{Pickands1975}. For a sequence of independent and identically (i.i.d.) distributed observations $\{x_i : i = 1, \ldots, n\}$, the excesses $x-u$ above a sufficiently high threshold \(u\), under mild conditions, can be approximated by a GPD (\cite{davison1990models}), whose cumulative distribution function (CDF) is given by
\begin{equation}
G(x \mid \xi,\sigma,u)=
\Pr(X \le x \,|\, X>u)=
\begin{cases}
1-\left(1+\dfrac{\xi(x-u)}{\sigma}\right)^{-1/\xi}, & \xi \neq 0,\\[2mm]
1-\exp\left(-\dfrac{x-u}{\sigma}\right), & \xi =0,
\end{cases}
\label{eq:GPD_cdf}
\end{equation}
for $u,\xi\in\mathbb{R}$ and $\sigma > 0
$, where the support is $x\geq u$ if $\xi\geq 0$ and $0\leq x\leq u-\sigma/\xi$ if $\xi<0$. Thus, the GPD is bounded if $\xi<0$ and unbounded above if $\xi\geq 0$.
In some applications, an additional implicit parameter is considered, namely the probability of being above the threshold, called the ``tail fraction'' and denoted by $\phi_u = \Pr(X > u)$. This quantity plays an important role in evaluating quantities such as the unconditional survival probability, and is given by
\begin{equation}
\Pr(X > x) = \phi_u \left[1 - \Pr(X \leq x \mid X > u)\right].
\label{eq:unconditional_survival}
\end{equation}
 \noindent 
The MLE of the tail fraction parameter $\phi_u$ is given by the sample proportion of observations above the threshold. The choice of threshold plays an important role in fitting the GPD, as inference varies across thresholds (\cite{tancredi2006accounting}).
\section{Methodology}\label{methodology}
\subsection{Proposed model}
\noindent In this section, we introduce a novel model called the EEVIMM. 
This model is flexible and can address common data issues, such as rainfall, life testing, financial trade durations, and insurance claims. 
Such data often exhibit instantaneous failures ("inliers at zero"), followed by early failures, which are clusters of observations near zero with high frequency (also called "inliers near zero"), and some with higher magnitudes and lower frequency. Standard statistical models may not be sufficiently flexible to capture all these characteristics simultaneously.
The model partitions the support of the non-negative random variable \(X\) into four distinct regions using two parameters \(d\) and \(u\), where \(0 < d < u\).
Let $\boldsymbol{\Theta} = (\phi_1, \phi_2, \Phi_1, d, \Phi_2, \phi_3, u, \xi, \sigma)$
denote the parameter vector, the CDF of the proposed EEVIMM is defined as
\begin{equation}
\label{eq:EEVIMM}
F(x \mid \phi_1, \phi_2, \Phi_1, d, \Phi_2, \phi_3, u, \xi, \sigma) =
\begin{cases}
0,
& \text{if } x < 0, \\[0.6em]
\phi_1,
& \text{if } x = 0, \\[0.6em]
\phi_1 + \phi_2
\displaystyle
\frac{F_1(x \mid \Phi_1)}
{F_1(d \mid \Phi_1)},
& \text{if } 0 < x < d, \\[1.0em]
\phi_1 + \phi_2
+ \left(1-\phi_1-\phi_2-\phi_3\right)
\displaystyle
\frac{F_2(x \mid \Phi_2)-F_2(d \mid \Phi_2)}
     {F_2(u \mid \Phi_2)-F_2(d \mid \Phi_2)},
& \text{if } d \le x < u, \\[1.3em]
(1-\phi_3) + \phi_3\, G(x \mid u,\xi,\sigma),
& \text{if } x \ge u ,
\end{cases}
\end{equation}
where \(0 < \phi_1, \phi_2, \phi_3 < 1\), \(\phi_1+\phi_2+\phi_3 < 1\), \(0<d<u\), \(\sigma>0\), and \(\xi \in \mathbb{R}\). 
The EEVIMM consists of four main components, each aiming to capture different features of the data distribution. Instantaneous failures are captured by the probability mass $\phi_1$ at $x = 0$, while early failures are modeled using an early-range distribution. Any suitable finite-range distribution $F_1(x \mid \Phi_1)$, such as the uniform, beta, or power distribution, may be used to capture this behaviour.
The observations above a sufficiently high threshold \(u\) are modeled using the tail distribution, namely the GPD, denoted by \(G(x \mid u,\xi,\sigma)\), where \(\xi\) controls the heaviness of the tail, \(\sigma > 0\) controls the spread, and \(u\) separates the tail from the rest.
For the intermediate range \(d \leq x < u\), the data are modeled using the bulk distribution. This component captures the main body of the data, excluding instantaneous and early-range events, and the tail distribution. The bulk component can be modeled using distributions such as the gamma, Weibull, beta, or lognormal. The model defined in (\ref{eq:EEVIMM}) is general and requires a suitable choice of the early-range and bulk distributions, depending on the nature of the data. For computational simplicity and stable parameter estimation, a suitable restriction may be imposed on the scale parameter of the bulk distribution, such as the gamma or Weibull distribution, to avoid identifiability issues.
We can express the model given in (\ref{eq:EEVIMM}) as a four-component mixture as
\[
f(x)
= \phi_1 \,\delta_0(x)
+ \phi_2 \, f_1^*(x \mid \Phi_1)
+ (1-\phi_1-\phi_2-\phi_3)\, f_2^*(x \mid \Phi_2)
+ \phi_3 \, g(x \mid u,\xi,\sigma),
\]
where 
\[
\begin{aligned}
f_1^*(x \mid \Phi_1)
&= \frac{f_1(x \mid \Phi_1)}{F_1(d \mid \Phi_1)}\,\mathbb{I}_{(0,d)}(x),
\qquad
f_2^*(x \mid \Phi_2)
= \frac{f_2(x \mid \Phi_2)}
{F_2(u \mid \Phi_2)-F_2(d \mid \Phi_2)}
\,\mathbb{I}_{[d,u)}(x).
\end{aligned}
\]
Here \(\delta_0(x)\) denotes the Dirac delta function at zero, \(f_1(x\mid\Phi_1)\) denotes the density function (DF) corresponding to the early-range distribution CDF \(F_1(x\mid\Phi_1)\), \(f_2(x\mid\Phi_2)\) denotes the DF corresponding to the bulk distribution \(F_2(x\mid\Phi_2)\), and \(g(x\mid u,\xi,\sigma)\) denotes the DF of the GPD \(G(x \mid u,\xi,\sigma)\).
The DF corresponding to the EEVIMM defined in (\ref{eq:EEVIMM}) can be written as
\begin{equation}
f(x \mid \Theta)=
\begin{cases}
\phi_1, & x = 0, \\[6pt]

\displaystyle 
\phi_2 \,
\frac{f_1(x \mid \Phi_1)}{F_1(d \mid \Phi_1)}, 
& 0 < x < d, \\[6pt]

\displaystyle
(1-\phi_1-\phi_2-\phi_3)\,
\frac{f_2(x \mid \Phi_2)}
     {F_2(u \mid \Phi_2)-F_2(d \mid \Phi_2)}, 
& d \le x < u, \\[8pt]

\displaystyle
\phi_3\, g(x \mid u,\xi,\sigma),
& x \ge u .
\end{cases}
\label{eq:EEVIMM_general_pdf}
\end{equation}

In this paper, we use the power distribution to model early failures, while the Weibull distribution with a unit scale parameter is used as the bulk distribution for simulation and parameter estimation.
So, the  DF of the EEVIMM for particular early and bulk distribution is given by
\begin{equation}
\small
\begin{aligned}
f(x)=
\begin{cases}
\phi_1, & x = 0,\\
\phi_2 \dfrac{\theta}{d^{\theta}} x^{\theta-1}, & 0 < x < d,\\
(1-\phi_1-\phi_2-\phi_3)
\dfrac{k x^{k-1} e^{-x^k}}
{e^{-d^k}-e^{-u^k}}, & d \le x < u,\\
\phi_3\, g(x\mid\xi,\sigma,u), & x \ge u,
\end{cases}
\qquad
g(x\mid\xi,\sigma,u)=
\begin{cases}
\dfrac{1}{\sigma}
\left(1+\dfrac{\xi(x-u)}{\sigma}\right)^{-1/\xi-1},
& \xi\neq0,\\[8pt]
\dfrac{1}{\sigma}
e^{-(x-u)/\sigma},
& \xi=0.
\end{cases}
\end{aligned}
\label{evimmpdf}
\end{equation}
here \(\theta > 0\) and \(k > 0\).

We study key properties of the proposed model, including continuity and differentiability at $d$ and $u$, as well as related risk measures.
\subsection{Continuity and differentiability conditions}
\noindent In this subsection, we establish the conditions under which the DF of the EEVIMM is continuous and differentiable at the
transition points $d$ and $u$.
\begin{proposition}[Continuity at $d$ and $u$]
\label{rem:Continuity}
The DF of EEVIMM~\eqref{eq:EEVIMM_general_pdf} is 
continuous at $d$ and the threshold $u$, that is,
\[
f(d^{-}) = f(d^{+})
\quad \text{and} \quad
f(u^{-}) = f(u^{+}),
\]
if and only if the following conditions hold:

\medskip
\noindent
\textbf{(i) Continuity at $d$:}
\[
\phi_2 \frac{f_1(d \mid \Phi_1)}{F_1(d \mid \Phi_1)}
=
\left(1-\phi_1-\phi_2-\phi_3\right)
\frac{f_2(d \mid \Phi_2)}
     {F_2(u \mid \Phi_2)-F_2(d \mid \Phi_2)}.
\]
For the specific choice of the early-range and bulk distribution in \eqref{evimmpdf}, the continuity condition at \(d\) reduces to
\[
\theta
=
\frac{d}{\phi_2}
(1-\phi_1-\phi_2-\phi_3)
\frac{
k
d^{k-1}
\exp\!\left[-d^k\right]
}{
\exp\!\left[-d^k\right]
-
\exp\!\left[-u^k\right]
}.
\]
\noindent
\medskip
\noindent
\textbf{(ii) Continuity at $u$:}
\[
\sigma
=
\frac{
\phi_3 \left[ F_2(u \mid \Phi_2)-F_2(d \mid \Phi_2) \right]
}{
\left(1-\phi_1-\phi_2-\phi_3\right)
f_2(u \mid \Phi_2)
}.
\]
For the specific choice of the bulk distribution in \eqref{evimmpdf}, the continuity condition at the threshold \(u\) reduces to
\[
\sigma
=
\frac{
\phi_3
\left[
\exp\!\left(-d^k\right)
-
\exp\!\left(-u^k\right)
\right]
}{
(1-\phi_1-\phi_2-\phi_3)
k
u^{k-1}
\exp\!\left[-u^k\right]
}.
\]
\noindent
These conditions ensure that the DF~\eqref{evimmpdf} is continuous at \(d\) and \(u\), providing smooth transitions between the early-range and bulk components at \(d\), and between the bulk and tail components at \(u\).
\end{proposition}
\begin{proposition}
[Differentiability at the thresholds $d$ and $u$]
\label{rem:differentiability}
The DF of EEVIMM~\eqref{eq:EEVIMM_general_pdf} is  
differentiable at $d$ and the threshold $u$, that is,
\[
f'(d^{-}) = f'(d^{+})
\quad \text{and} \quad
f'(u^{-}) = f'(u^{+}),
\]
if and only if the continuity conditions in Proposition~\ref{rem:Continuity} hold together with the following conditions:\\
\medskip
\noindent
\textbf{(i) Differentiability at $d$:}
\[
\phi_2 \frac{f_1'(d \mid \Phi_1)}{F_1(d \mid \Phi_1)}
=
\left(1-\phi_1-\phi_2-\phi_3\right)
\frac{f_2'(d \mid \Phi_2)}
     {F_2(u \mid \Phi_2)-F_2(d \mid \Phi_2)}.
\]
For the specific choice of the early-range and bulk distributions in \eqref{evimmpdf}, the differentiability condition at \(d\) reduces to
\[
\phi_2 \frac{\theta(\theta-1)}{d^2}
=
\left(1-\phi_1-\phi_2-\phi_3\right)
\frac{
f_2(d \mid k)
\left(
\frac{k-1}{d}
-
k d^{k-1}
\right)
}{
\exp\!\left[-d^k\right]
-
\exp\!\left[-u^k\right]
}.
\]
\medskip
\noindent
\textbf{(ii) Differentiability at $u$:}
\[
\sigma^2
=
-\frac{
\phi_3(1+\xi)\{F_2(u|\Phi_2)-F_2(d|\Phi_2)\}
}{
(1-\phi_1-\phi_2-\phi_3)f_2'(u|\Phi_2)
}.
\]
For the specific choice of the bulk distribution in \eqref{evimmpdf}, the differentiability condition at the threshold \(u\) reduces to
\[
\sigma^2
=
-\frac{
\phi_3(1+\xi)
\left(
\exp[-d^k]-\exp[-u^k]
\right)
}{
\left(1-\phi_1-\phi_2-\phi_3\right)
f_2(u \mid k)
\left(
\frac{k-1}{u}
-
k u^{k-1}
\right)
}.
\]
\noindent
These conditions ensure that the DF of the EEVIMM is differentiable at \(d\) and \(u\), providing smooth transitions between the model components at \(d\) and \(u\).
\end{proposition}

\subsection{Risk measures}\label{risk}
\noindent Risk measures are tools used to quantify the level of risk associated with a random variable by summarising it into a single value. Several risk measures have been developed in the literature; for details, see \cite{szego2004risk}.
Value-at-Risk (VaR) and Tail-Value-at-Risk (TVaR) are among the most widely used risk measures in risk analysis. These measures are particularly useful for heavy-tailed distributions when information above a given threshold is required. In this section, we derive the VaR and TVaR for the EEVIMM. 
\begin{proposition}[Risk measure: VaR]
\label{prop:var}
Assume that $F_1(\cdot\mid\Phi_1)$,
$F_2(\cdot\mid\Phi_2)$ and
$G(\cdot\mid \xi, \sigma, u)$ admit inverse functions. 
For a probability level $p\in(0,1)$, the VaR denoted by $\mathrm{VaR}_p$, is defined as $\mathrm{VaR}_p = \inf\{x: F(x)\ge p\}$.
Then the $\mathrm{VaR}_p$ of \eqref{eq:EEVIMM} is given by
\[
\mathrm{VaR}_p=
\begin{cases}
0,
& 0 \le p \le \phi_1, \\[10pt]

F_1^{-1}\!\left(
\dfrac{p-\phi_1}{\phi_2} F_1(d\mid\Phi_1)
\mid\Phi_1
\right),
& \phi_1 < p \le \phi_1+\phi_2, \\[12pt]

F_2^{-1}\!\left(
F_2(d\mid\Phi_2)
+
\dfrac{p-(\phi_1+\phi_2)}{1-\phi_1-\phi_2-\phi_3}
\left(
F_2(u\mid\Phi_2)
-
F_2(d\mid\Phi_2)
\right)
\right),
& \phi_1+\phi_2 < p \le 1-\phi_3, \\[14pt]

G^{-1}\!\left(
\dfrac{p-(1-\phi_3)}{\phi_3}
\mid \xi, \sigma, u
\right),
& 1-\phi_3 < p \le 1.
\end{cases}
\]
\end{proposition}
\begin{proposition}[Risk measure: TVaR]
\label{prop:tvar}
For a given probability level $p$, let $\mathrm{VaR}_p$ be defined as in 
Proposition~\ref{prop:var}. 
The \emph{TVaR} also called Expected Shortfall or conditional tail expectation, is defined for a random variable $X$ as
\[
\mathrm{TVaR}_p(X)
= \mathbb{E}[X \mid X > \mathrm{VaR}_p]
= \frac{\mathbb{E}\!\left[ X \,\mathbf{1}_{\{X > \mathrm{VaR}_p\}} \right]}
       {\mathbb{P}(X > \mathrm{VaR}_p)}
= \frac{1}{1-p} \int_{\mathrm{VaR}_p}^{\infty} x\, f(x)\,dx,
\]
where $f(x)$ is the DF defined in
\eqref{eq:EEVIMM_general_pdf}.
The TVaR at level~$p$ can therefore be expressed as
\begin{equation}
\resizebox{1.11\textwidth}{!}{$
\mathrm{TVaR}_p =
\begin{cases}
\displaystyle
\frac{1}{1-p}\Biggl[
\int_{0^+}^{d} x \,
\phi_2 \frac{f_1(x\mid\Phi_1)}{F_1(d\mid\Phi_1)}\,dx
+
\int_{d}^{u} x \,
\frac{1-\phi_1-\phi_2-\phi_3}
     {F_2(u\mid\Phi_2)-F_2(d\mid\Phi_2)}
     f_2(x\mid\Phi_2)\,dx
+
\int_{u}^{\infty} x\, \phi_3\, g(x\mid \xi,\sigma, u)\,dx
\Biggr],
& 0 \le p \le \phi_1,\\[18pt]

\displaystyle
\frac{1}{1-p}\Biggl[
\int_{\mathrm{VaR}_p}^{d} x \,
\phi_2 \frac{f_1(x\mid\Phi_1)}{F_1(d\mid\Phi_1)}\,dx
+
\int_{d}^{u} x \,
\frac{1-\phi_1-\phi_2-\phi_3}
     {F_2(u\mid\Phi_2)-F_2(d\mid\Phi_2)}
     f_2(x\mid\Phi_2)\,dx
+
\int_{u}^{\infty} x\, \phi_3\, g(x\mid \xi, \sigma, u)\,dx
\Biggr],
& \phi_1 < p \le \phi_1+\phi_2,\\[18pt]

\displaystyle
\frac{1}{1-p}\Biggl[
\int_{\mathrm{VaR}_p}^{u} x \,
\frac{1-\phi_1-\phi_2-\phi_3}
     {F_2(u\mid\Phi_2)-F_2(d\mid\Phi_2)}
     f_2(x\mid\Phi_2)\,dx
+
\int_{u}^{\infty} x\, \phi_3\, g(x\mid \xi, \sigma, u)\,dx
\Biggr],
& \phi_1+\phi_2 < p \le 1-\phi_3,\\[18pt]

\displaystyle
\frac{1}{1-p}
\int_{\mathrm{VaR}_p}^{\infty} x \,\phi_3\, g(x\mid \xi, \sigma, u)\,dx,
& 1-\phi_3 < p \le 1.
\end{cases}
$}
\end{equation}
\noindent
Note the lower limit $0^+$ in the first case excludes the possible atom at $x=0$ of mass $\phi_1$; when $p\le\phi_1$ we have $\mathrm{VaR}_p=0$. For the EEVIMM, these quantities can be obtained analytically through the derived integrals or estimated numerically using a Monte Carlo simulation.
\end{proposition}
The proofs of Propositions~\ref{rem:Continuity} and \ref{rem:differentiability} are provided in Appendix~A, whereas Propositions~\ref{prop:var} and \ref{prop:tvar} follow directly from the definitions of the corresponding risk measures; see \citet{szego2004risk}.

\subsection{Algorithm for data simulation}\label{algorithm}
\noindent In this section, we present an algorithm for simulating data $x_i$, 
$i = 1, 2, 3, \dots, n$, from the proposed model, EEVIMM, with the CDF 
$F(x \mid \bm{\Theta})$
(\ref{eq:EEVIMM}). The simulation algorithm is given in Algorithm~\ref{alg:EEVIMM_general}.
\begin{algorithm}[t!]
\caption{Simulation algorithm for the EEVIMM}
\label{alg:EEVIMM_general}
\begin{algorithmic}[1]
\Require $n$, $\phi_1, \phi_2, \Phi_1, d, \Phi_2, \phi_3, u, \xi, \sigma$
\Ensure A random sample $\{x_1, x_2, \dots, x_n\}$ from the EEVIMM.
\State Set $\phi_4 \gets 1 - \phi_1 - \phi_2 - \phi_3$
\State Start with an empty vector $x$ of length $n$.

\For{$i = 1$ to $n$}

    \State Generate $U \sim \text{Uniform}(0,1)$

    \If{$U \le \phi_1$}
        \State $x_i \gets 0$
        \Comment{Mass at zero}

    \ElsIf{$\phi_1 < U \le \phi_1 + \phi_2$}
        \State $U_1 \gets \dfrac{U - \phi_1}{\phi_2}$
        \Comment{Rescale to $(0,1)$}
        \State $x_i \gets F_1^{-1}\!\left(U_1\,F_1(d \mid \Phi_1) \mid \Phi_1\right)$
        \Comment{Quantile of early-range distribution}

    \ElsIf{$\phi_1 + \phi_2 < U \le 1 - \phi_3$}
        \State $U_2 \gets \dfrac{U - (\phi_1 + \phi_2)}{\phi_4}$
        \Comment{Rescale to $(0,1)$}
        \State Compute $F_{2d} \gets F_2(d \mid \Phi_2)$
        \Comment{CDF of $F_2$ at $d$}
        \State Compute $F_{2u} \gets F_2(u \mid \Phi_2)$
        \Comment{CDF of $F_2$ at $u$}
        \State $p \gets F_{2d} + U_2 (F_{2u} - F_{2d})$
        \State $x_i \gets F_2^{-1}(p \mid \Phi_2)$
        \Comment{Quantile of bulk distribution}

    \Else
        \State $U_3 \gets \dfrac{U - (1 - \phi_3)}{\phi_3}$
        \Comment{Rescale to $(0,1)$}

        \If{$\xi \neq 0$}
            \State $x_i \gets u + \dfrac{\sigma}{\xi}
            \left[(1 - U_3)^{-\xi} - 1\right]$
        \Else
            \State $x_i \gets u - \sigma \log(1 - U_3)$
        \EndIf
        \Comment{Quantile of GPD}
    \EndIf

\EndFor

\State \Return $\{x_1, x_2, \dots, x_n\}$

\end{algorithmic}
\end{algorithm}
\subsection{Parameter estimation}\label{estimation}
\noindent In this section, we present the likelihood function of the proposed EEVIMM and describe the corresponding estimation method and algorithm based on maximum likelihood estimation.
Parameter estimation is performed numerically due to the absence of closed-form solutions. The estimation methodology and the associated numerical algorithm are described in the following subsections.
\subsubsection{Maximum likelihood estimation}\label{mle}
\noindent Let \( \bm{x} = (x_1, x_2, \dots, x_n) \) be a random sample from the EEVIMM $f(x\mid \bm{\Theta})$ \eqref{eq:EEVIMM_general_pdf}. The likelihood function is 
$L(\bm{x}; \bm{\Theta}) = \prod_{i=1}^n f(x_i; \bm{\Theta})$,
where \(f(x_i; \bm{\Theta})\) represents the DF. We introduce indicator functions corresponding to distinct regions of the data, as
\[
I_0(x_i) =
\begin{cases} 
1, & \text{if } x_i = 0, \\
0, & \text{otherwise},
\end{cases}
\quad
I_1(x_i; d) =
\begin{cases} 
1, & \text{if } 0 < x_i < d, \\
0, & \text{otherwise},
\end{cases}
\]
\[
I_2(x_i; d,u) =
\begin{cases} 
1, & \text{if } d \le x_i < u, \\
0, & \text{otherwise},
\end{cases}
\quad \text{and} \quad
I_3(x_i; u) =
\begin{cases} 
1, & \text{if } x_i \ge u, \\
0, & \text{otherwise}.
\end{cases}
\]
\noindent
The log-likelihood function is given by $
\ell(\boldsymbol{\Theta} \mid \bm{x})
:= \log \mathcal{L}(\boldsymbol{\Theta} \mid \bm{x})$,
\begin{equation}
\begin{aligned}
\ell(\boldsymbol{\Theta}\mid\bm{x})
&=
n_0\log\phi_1
+ n_1\log\phi_2
+ n_2\log(1-\phi_1-\phi_2-\phi_3)
+ n_3\log\phi_3
\\
&\quad
- n_2\log\!\left[F_2(u\mid\Phi_2)-F_2(d\mid\Phi_2)\right]
+ \sum_{i=1}^n I_1(x_i;d) \log f_1(x_i\mid\Phi_1) - n_1 \log F_1(d\mid\Phi_1)
\\
&\quad
+ \sum_{i=1}^n I_2(x_i;d,u) \log f_2(x_i\mid\Phi_2)
+\sum_{i=1}^n I_3(x_i;u) \log g(x_i\mid u,\xi,\sigma),
\end{aligned}
\end{equation}
$n_0 = \sum_{i=1}^n I_0(x_i),\ 
n_1 = \sum_{i=1}^n I_1(x_i; d),\ 
n_2 = \sum_{i=1}^n I_2(x_i; d,u),\ 
n_3 = \sum_{i=1}^n I_3(x_i; u)$.
For the specific choice of the early-range and bulk distributions in \eqref{evimmpdf}, the log-likelihood function reduces to
\begin{equation}
\begin{aligned}
\ell(\boldsymbol{\Theta}\mid\bm{x})
&=
n_0 \log \phi_1
+ n_1 \log \phi_2
+ n_1 \log \theta
- n_1 \theta \log d
 +
 (\theta-1)\sum_{i=1}^n I_1(x_i;d)\log x_i
+ n_2 \log(1-\phi_1-\phi_2-\phi_3)
\\
&\quad
- n_2 \log (\exp(-d^k) - \exp(-u^k))
+ \sum_{i=1}^n I_2(x_i;d,u)
\left[
\log k
+ (k-1)\log x_i
- x_i^k
\right]
+
 n_3 \log \phi_3
- n_3 \log \sigma
\\
&\quad
- \left(\frac{1}{\xi}+1\right)
\sum_{i=1}^n I_3(x_i;u)
\log\!\left(1+\xi\frac{x_i-u}{\sigma}\right),
\qquad \xi \neq 0 , ~~\text{and}
\end{aligned}
\end{equation}
\begin{equation}
\begin{aligned}
\ell(\boldsymbol{\Theta}\mid\bm{x})
&=
n_0 \log \phi_1
+ n_1 \log \phi_2
+ n_1 \log \theta
- n_1 \theta \log d
+ (\theta-1)\sum_{i=1}^n I_1(x_i;d)\log x_i
+ n_2 \log(1-\phi_1-\phi_2-\phi_3)
\\
&\quad
- n_2 \log (\exp(-d^k) - \exp(-u^k))
+ \sum_{i=1}^n I_2(x_i;d,u)
\left[
\log k
+ (k-1)\log x_i
- x_i^k
\right]
\\
&\quad
+ n_3 \log \phi_3
- n_3 \log \sigma
- \frac{1}{\sigma}\sum_{i=1}^n I_3(x_i;u)(x_i-u),
\qquad \xi = 0 .
\end{aligned}
\end{equation}
The MLEs of the unknown parameter vector 
\(\boldsymbol{\Theta}= (\phi_1, \phi_2, \theta, d, k, \phi_3, u, \xi, \sigma)\) 
are obtained by maximising the log-likelihood function.
The likelihood function of the EEVIMM is often complex and lacks a closed-form solution, so numerical methods are required to compute the MLEs. 
Optimisation of the likelihood for such EVMMs is often challenging due to the multimodal nature of the likelihood surface. 
In general, the DF of an EVMM need not be continuous at the threshold, which further complicates the optimisation procedure; see \cite{dey2016extreme} for more details.

\subsubsection{Algorithm for parameter estimation}
\noindent This section presents the estimation procedure for estimating the parameter vector \(\Theta\). Direct optimisation of the likelihood for such EVMMs is often challenging due to the multimodal nature of the likelihood surface. Another challenge is that optimisation is very sensitive to the initial threshold values. Various advanced optimisation strategies, such as multiple starting values or stochastic optimisation methods, may also be used to improve numerical stability and reduce the risk of the optimisation routine getting stuck in a local maximum. To overcome the optimisation issues associated with the multimodality of the likelihood surface, a profile-likelihood approach combined with a grid search over candidate threshold values is a better alternative approach, and strongly recommended in \cite{dey2016extreme} and \cite{hu2018evmix}. We use the profile likelihood with a grid search to estimate the threshold parameter \(u\), while the parameter \(d\) is selected using an appropriate GoF measure over the candidate grid.

Optimisation is performed using the L-BFGS-B numerical optimisation algorithm via the \texttt{optim} function from the \texttt{stats} package in \cite{team2020ra}, with a maximum of 20{,}000 iterations to prevent premature termination. Model-based quantile estimation is performed only when the likelihood optimisation algorithm has converged successfully. Appropriate box constraints or user-specified lower and upper bounds are imposed on the parameters to ensure numerical stability. The overall computational procedure is summarised in Algorithm~\ref{alg:EEVIMM_estimation}. This estimation algorithm is used in the simulation study and real data analysis presented in Sections~\ref{SStudy} and~\ref{Applications}.

\begin{algorithm}[!t]
\caption{Estimation algorithm for the EEVIMM}
\label{alg:EEVIMM_estimation}
\begin{algorithmic}[1]

\State \textbf{Step 1.}
Given $\{x_1,\ldots,x_n\}$, construct a candidate grid
$\mathcal{D}=\{d_1,\ldots,d_m\}$ from empirical quantiles of the data. For each candidate value $d$, let $\mathcal{X}_d=\{x_i:0<x_i<d\}$. The estimate $\hat d$ is selected using a GoF measure, such as minimising the Kolmogorov-Smirnov distance
\[
\hat d=\arg\min_{d\in\mathcal{D}}
\sup_{x\in(0,d)}
\left|
\widehat F_{\mathcal{X}_d}(x)-
\frac{F_1(x\mid\Phi_1)}{F_1(d\mid\Phi_1)}
\right|,
\]
where $\widehat F_{\mathcal{X}_d}$ denotes the empirical CDF of $\mathcal{X}_d$. Obtain initial values $\Phi_1^{(0)}$ from $\mathcal{X}_{\hat d}$ using moment-based estimators or MLE, depending on the chosen form of $F_1(\cdot)$.

\State \textbf{Step 2.}
Let $\mathcal{U}=\{u_1,\ldots,u_L\}$ be a grid of threshold candidates selected from high empirical quantiles, satisfying $u>\hat d$ and
$\sum_{i=1}^n I(x_i\ge u)\ge m_0$, where $m_0$ denotes the minimum number of exceedances.
\State \textbf{Step 3.}
For each $u\in\mathcal{U}$, estimate the mixing probabilities either using sample proportions or their closed-form MLEs,
\[
\hat{\phi}_1=\frac{n_0}{n},\qquad
\hat{\phi}_2=\frac{n_1}{n},\qquad
\hat{\phi}_3(u)=\frac{n_3}{n}.
\]
Evaluate the profile log-likelihood
\[
\ell_p(u)=\max_{(\Phi_1,\Phi_2,\xi,\sigma)}
\ell(\Phi_1,\Phi_2,\xi,\sigma;\hat d,u,\hat\phi_1,\hat\phi_2,\hat\phi_3(u)).
\]
\State \textbf{Step 4.}
The threshold estimate $\hat u$ is obtained by maximising the profile
log-likelihood over the candidate grid $\mathcal{U}$, i.e.,
\[
\hat u = \arg\max_{u \in \mathcal{U}} \, \ell_p(u).
\]

\State \textbf{Step 5.}
Obtain initial values $(\xi^{(0)},\sigma^{(0)})$ by fitting a GPD to the exceedances $\{x_i:x_i\ge \hat u\}$, and obtain $\Phi_2^{(0)}$ using observations in $\{x_i:\hat d\le x_i<\hat u\}$.
\State \textbf{Step 6.}
Given $(\hat d,\hat u)$, the mixing probabilities
$(\hat\phi_1,\hat\phi_2,\hat\phi_3)$ are computed again using MLE or the sample-proportion
estimators defined in Step~3, evaluated at $u=\hat u$.
\State \textbf{Step 7.}
Finally, given $(\hat d,\hat u)$ and $(\hat\phi_1,\hat\phi_2,\hat\phi_3)$, estimate the remaining parameters $(\Phi_1,\Phi_2,\xi,\sigma)$ by maximising the log-likelihood
\[
(\hat\Phi_1,\hat\Phi_2,\hat\xi,\hat\sigma)
=
\arg\max_{(\Phi_1,\Phi_2,\xi,\sigma)\in\mathcal{B}}
\ell(\Phi_1,\Phi_2,\xi,\sigma;\hat d,\hat u,\hat\phi_1,\hat\phi_2,\hat\phi_3),
\]
using the L-BFGS-B algorithm with initial values
$(\Phi_1^{(0)},\Phi_2^{(0)},\xi^{(0)},\sigma^{(0)})$
obtained from Steps~1 and~5, where $\mathcal{B}$ denotes user-specified lower and upper bounds imposed for numerical stability. The resulting estimator is
\[
\hat{\boldsymbol{\Theta}}
=
(\hat\phi_1,\hat\phi_2,\hat\Phi_1,\hat d,\hat\Phi_2,\hat\phi_3,\hat u,\hat\xi,\hat\sigma).
\]
\end{algorithmic}
\end{algorithm}

\subsection{Asymptotic distribution of the MLE}
In this Subsection, we establish the asymptotic properties of the MLEs of the proposed model parameters. The derivation of the asymptotic distribution is carried out under suitable regularity conditions and assumptions on the threshold and tail parameters.

\begin{remark}[Assumption: threshold (\(u\)) and paramter (\(d\))]
For the development of the asymptotic theory, the threshold \(u\) and cut-point \(d\) are assumed to be known and fixed at \(u=u_0\) and \(d=d_0\). Treating them as unknown parameters may violate the regularity conditions required for the asymptotic distribution of the estimators. Accordingly, the score functions and Fisher information matrix are derived under this assumption. In practice, \(u\) may be estimated using profile likelihood over a grid of candidate thresholds, while \(d\) may be selected using a GoF measure, profile likelihood, or model selection criteria. Once selected, \(u_0\) and \(d_0\) are treated as fixed in the subsequent inference.
\end{remark}

\begin{remark}[Regularity assumptions for the GPD shape parameter (\(\xi\))]
The asymptotic normality of the MLE of \(\xi\) holds under suitable regularity conditions, including finite Fisher information. The MLE of the $\xi$ does not satisfy the regularity conditions when $\xi \in (-1,-0.5)$; moreover, MLEs do not exist for $\xi < -1$, as the likelihood function can become unbounded (~\citet{castillo1997fitting,do2012semiparametric}).
It is common practice to assume that $\xi$ lies within the interval $-0.5 < \xi < 0.5$, since this range is frequently encountered in practical applications as noted in~\citet{hosking1985estimation}.
\end{remark}

\noindent \textbf{Notation and definitions:} \\
For \(X \geq u_{0}\), let \(X \sim \mathrm{GPD}(\xi, \sigma, u_{0})\), and define the standardized variable \(Z = \frac{X-u_{0}}{\sigma} \sim \mathrm{GPD}(\xi,1,0)\), where \(\mathbb{E}_Z[\cdot]\) denotes expectation with respect to \(Z\). For simplicity, let \(\phi_4 = 1-\phi_1-\phi_2-\phi_3\).
\subsubsection{Score Functions: Case \texorpdfstring{\(\xi \ne 0\)}{xi != 0}}
\noindent Using the EEVIMM DF defined in ~\eqref{evimmpdf}, the log-likelihood function for a single observation \(x\) is given by
\begin{equation}
\ell(\boldsymbol{\Theta}\mid x)=
\begin{cases}
\log(\phi_1), & x = 0, \\[8pt]

\log(\phi_2) + \log \theta - \theta \log d_0 + (\theta - 1)\log x,
& 0 < x < d_0, \\[10pt]

\log(1 - \phi_1 - \phi_2 - \phi_3)
+ \log k
+ (k - 1)\log x
- x^k
- \log\!\left( e^{-d_0^k} - e^{-u_0^k} \right),
& d_0 \le x < u_0, \\[12pt]

\log(\phi_3)
- \log \sigma
-
\left(\frac{1}{\xi} + 1\right)
\log\!\left(1 + \frac{\xi(x-u_0)}{\sigma}\right),
& x \ge u_0,\ \xi\neq0, \\[12pt]

\log(\phi_3)
-\log\sigma
-\dfrac{x-u_0}{\sigma},
& x\ge u_0,\ \xi=0.
\end{cases}
\label{loglikelihood_eevimm}
\end{equation}
\noindent Let 
$\boldsymbol{\Theta}^{*} = (\phi_1, \phi_2, \theta, k, \phi_3, \xi, \sigma)^\top$ be the parameter vector. The corresponding score functions are
\begin{equation*}
\small
\frac{\partial \ell}{\partial \phi_1} =
\begin{cases}
\dfrac{1}{\phi_1}, & x = 0, \\
-\dfrac{1}{\phi_4}, & d_0 \le x < u_0, \\
0, & \text{otherwise},
\end{cases}
\quad
\frac{\partial \ell}{\partial \phi_2} =
\begin{cases}
\dfrac{1}{\phi_2}, & 0 < x < d_0, \\
-\dfrac{1}{\phi_4}, & d_0 \le x < u_0, \\
0, & \text{otherwise},
\end{cases}
\quad
\frac{\partial \ell}{\partial \phi_3} =
\begin{cases}
-\dfrac{1}{\phi_4}, & d_0 \le x < u_0, \\
\dfrac{1}{\phi_3}, & x \ge u_0, \\
0, & \text{otherwise}.
\end{cases}
\end{equation*}

\begin{equation*}
\small
\frac{\partial \ell}{\partial \theta} =
\begin{cases}
\dfrac{1}{\theta}-\log d_0+\log x, & 0<x<d_0,\\[4pt]
0, & \text{otherwise},
\end{cases}
\quad\quad
\frac{\partial \ell}{\partial k} =
\begin{cases}
\dfrac{1}{k}+\log x-x^k\log x
+\dfrac{d_0^k\log d_0\,e^{-d_0^k}-u_0^k\log u_0\,e^{-u_0^k}}
{e^{-d_0^k}-e^{-u_0^k}},
& d_0\le x<u_0,\\[8pt]
0, & \text{otherwise}.
\end{cases}
\end{equation*}
\begin{equation*}
\frac{\partial \ell}{\partial \sigma} =
\begin{cases}
-\dfrac{1}{\sigma}
+
\left(\dfrac{1}{\xi}+1\right)
\dfrac{\xi Z}{\sigma(1+\xi Z)},
& x \ge u_0, \\[8pt]
0, & \text{otherwise},
\end{cases}
\quad\quad
\frac{\partial \ell}{\partial \xi} =
\begin{cases}
\dfrac{1}{\xi^2}\log(1+\xi Z)
-\left(\dfrac{1}{\xi}+1\right)
\dfrac{Z}{1+\xi Z},
& x \ge u_0, \\[8pt]
0, & \text{otherwise}.
\end{cases}
\end{equation*}
One can verify that
\begin{align*}
&\mathbb{E}\!\left[\frac{\partial}{\partial \phi_1}\ell(X)\right]=0,
\quad
\mathbb{E}\!\left[\frac{\partial}{\partial \phi_2}\ell(X)\right]=0,
\quad
\mathbb{E}\!\left[\frac{\partial}{\partial \phi_3}\ell(X)\right]=0,
\\[6pt]
&\mathbb{E}\!\left[\frac{\partial}{\partial \theta}\ell(X)\right]=0,
\quad
\mathbb{E}\!\left[\frac{\partial}{\partial k}\ell(X)\right]=0,
\quad
\mathbb{E}\!\left[\frac{\partial}{\partial \sigma}\ell(X)\right]=0,
\quad
\mathbb{E}\!\left[\frac{\partial}{\partial \xi}\ell(X)\right]=0,
\end{align*}
\noindent The Fisher information matrix and its \((i,j)\)-th element are given by
\begin{align}
\mathcal{I}(\boldsymbol{\Theta}^{*})
&=
- \mathbb{E}\!\left[
\frac{\partial^2 \ell(\boldsymbol{\Theta}^{*}\mid \mathbf{x})}
{\partial \boldsymbol{\Theta}^{*}\,\partial (\boldsymbol{\Theta}^{*})^{T}}
\right],
\quad \text{and} \quad
\mathcal{I}_{ij}(\boldsymbol{\Theta}^{*})
=
- \mathbb{E}\!\left[
\frac{\partial^2 \ell(\boldsymbol{\Theta}^{*}\mid \bm{x})}
{\partial \Theta^{*}_i\,\partial \Theta^{*}_j}
\right].
\end{align}
where $\Theta_i^{*}$ denotes the $i$-th component of $\boldsymbol{\Theta}^{*}$. So, the elements of the Fisher information matrix are
\begin{align*}
\mathcal{I}_{\phi_1\phi_1}
&=
\frac{1}{\phi_1}+\frac{1}{\phi_4},
\qquad
\mathcal{I}_{\phi_2\phi_2}
=
\frac{1}{\phi_2}+\frac{1}{\phi_4},
\qquad
\mathcal{I}_{\phi_3\phi_3}
=
\frac{1}{\phi_3}+\frac{1}{\phi_4},
\qquad
\mathcal{I}_{\phi_1\phi_2}
=
\mathcal{I}_{\phi_1\phi_3}
=
\mathcal{I}_{\phi_2\phi_3}
=
\frac{1}{\phi_4},
\qquad
\mathcal{I}_{\theta\theta} = \frac{\phi_2}{\theta^2}.
\end{align*}
\begin{align*}
\mathcal{I}_{kk}
&=
\phi_4
\left[
\frac{1}{k^2}
+
\mathbb{E}_{X\in[d_0,u_0)}
\left\{X^k(\log X)^2\right\}
-
\frac{\partial^2}{\partial k^2}
\log\!\left(e^{-d_0^k}-e^{-u_0^k}\right)
\right].
\end{align*}
\begin{align*}
\mathcal{I}_{\sigma\sigma}
&=
\frac{\phi_3}{\sigma^2}
\left[
\left(\frac{1}{\xi}+1\right)
\mathbb{E}_Z
\left(
\frac{\xi Z(2+\xi Z)}{(1+\xi Z)^2}
\right)
-
1
\right],
\qquad
\mathcal I_{\sigma\xi}
=
\phi_3\,\mathbb{E}_Z
\left[
\dfrac{Z}{\sigma \xi (1+\xi Z)}
-
\left(\dfrac{1}{\xi}+1\right)
\dfrac{Z}{\sigma(1+\xi Z)^2}
\right].
\end{align*}

\begin{align*}
\mathcal I_{\xi\xi}
=
\phi_3\,
\mathbb{E}_Z
\left[
\frac{2}{\xi^3}\log(1+\xi Z)
-
\frac{2}{\xi^2}\frac{Z}{1+\xi Z}
-
\left(\frac{1}{\xi}+1\right)
\frac{Z^2}{(1+\xi Z)^2}
\right],
\qquad
\mathcal{I}_{\phi_1\theta}
=
\mathcal{I}_{\phi_1 k}
=
\mathcal{I}_{\phi_1 \xi}
=
\mathcal{I}_{\phi_1 \sigma}
=
\mathcal{I}_{\phi_2\theta}
=
0.
\end{align*}

\[
\mathcal{I}_{\phi_2 k}
=
\mathcal{I}_{\phi_2 \xi}
=
\mathcal{I}_{\phi_2 \sigma}
=
\mathcal{I}_{\phi_3\theta}
=
\mathcal{I}_{\phi_3 k}
=
\mathcal{I}_{\phi_3 \xi}
=
\mathcal{I}_{\phi_3 \sigma}
=
\mathcal I_{\theta k}
=
\mathcal I_{\theta\xi}
=
\mathcal I_{\theta\sigma}
=
\mathcal I_{k\xi}
=
\mathcal I_{k\sigma}
=
0.
\]
\begin{theorem}[Asymptotic normality of the MLE for $\xi \neq 0$]
Let \( X_1, X_2, \dots, X_n \) be i.i.d. observations from the EEVIMM with known \( d_0>0 \), \( u_0>d_0 \), and parameter vector
\[
\boldsymbol{\Theta}^{*} = (\phi_1, \phi_2, \theta, k, \phi_3, \xi, \sigma)^\top
\in (0,1)^2 \times (0,\infty)^2 \times (0,1) \times \left( (-0.5,\infty)\setminus\{0\} \right) \times (0,\infty),
\]
with \(0<\phi_1,\phi_2,\phi_3<1\) and \(\phi_1+\phi_2+\phi_3<1\). 
Suppose that the true parameter vector \( \boldsymbol{\Theta}^{*}_0 \) belongs to the interior of the parameter space and that the usual regularity conditions for maximum likelihood estimation are satisfied.
Let \( \widehat{\boldsymbol{\Theta}^{*}}_n \) be the MLE of $\boldsymbol{\Theta}^{*}_0$ based on the sample of size \( n \). Then, as \( n \to \infty \),
\[
\sqrt{n} \left( \widehat{\boldsymbol{\Theta}^{*}}_n - \boldsymbol{\Theta}^{*}_0 \right) 
\xrightarrow{d} \mathcal{N}_7 \left( \bm{0}, \, \mathcal{I}^{-1}(\boldsymbol{\Theta}^{*}_0) \right),
\]
where $\mathcal{I}(\boldsymbol{\Theta}^{*}_0)$ denotes the Fisher information matrix corresponding to a single observation, evaluated at the true parameter value \( \boldsymbol{\Theta}^{*}_0 \), and having entries $\mathcal{I}_{ij}(\boldsymbol{\Theta}^{*}_0)$.
The Fisher information matrix is given by
\[
\mathcal{I}(\boldsymbol{\Theta}^{*}_0) =
\begin{bmatrix}
\mathcal{I}_{\phi_1\phi_1} & \mathcal{I}_{\phi_1\phi_2} & 0 & 0 & \mathcal{I}_{\phi_1\phi_3} & 0 & 0 \\[6pt]
\mathcal{I}_{\phi_2\phi_1} & \mathcal{I}_{\phi_2\phi_2} & 0 & 0 & \mathcal{I}_{\phi_2\phi_3} & 0 & 0 \\[6pt]
0 & 0 & \mathcal{I}_{\theta\theta} & 0 & 0 & 0 & 0 \\[6pt]
0 & 0 & 0 & \mathcal{I}_{kk} & 0 & 0 & 0 \\[6pt]
\mathcal{I}_{\phi_3\phi_1} & \mathcal{I}_{\phi_3\phi_2} & 0 & 0 & \mathcal{I}_{\phi_3\phi_3} & 0 & 0 \\[6pt]
0 & 0 & 0 & 0 & 0 & \mathcal{I}_{\xi\xi} & \mathcal{I}_{\xi\sigma} \\[6pt]
0 & 0 & 0 & 0 & 0 & \mathcal{I}_{\sigma\xi} & \mathcal{I}_{\sigma\sigma}
\end{bmatrix}.
\]
\end{theorem}

\subsubsection{Score Functions: Case \texorpdfstring{\(\xi = 0\)}{xi = 0}}
\noindent Let the parameter vector be 
$\boldsymbol{\Theta}^{*}=(\phi_1,\phi_2,\phi_3,\theta,k,\sigma)$.  
Using the log-likelihood function in~\eqref{loglikelihood_eevimm}, the tail distribution reduces to an exponential distribution with mean $\sigma$. The corresponding derivatives with respect to \(\sigma\) are
\begin{align*}
\frac{\partial \ell}{\partial \sigma}
&=
\begin{cases}
-\dfrac{1}{\sigma} + \dfrac{x-u_0}{\sigma^2},
& x \ge u_0, \\[8pt]
0, & \text{otherwise},
\end{cases}
\qquad
\frac{\partial^2 \ell}{\partial \sigma^2}
=
\begin{cases}
\dfrac{1}{\sigma^2}
-\dfrac{2(x-u_0)}{\sigma^3},
& x \ge u_0, \\[8pt]
0, & \text{otherwise}.
\end{cases}
\end{align*}
\noindent For \(X \ge u_0\), we have \(Y = X - u_0 \sim \mathrm{Exp}(\sigma)\) with \(\mathbb{E}[Y] = \sigma\).  
The Fisher information entries involving \(\sigma\) are
$\mathcal{I}_{\sigma\sigma} = \frac{\phi_{3}}{\sigma^{2}}, 
\quad
\mathcal{I}_{\phi_1\sigma}
=
\mathcal{I}_{\phi_2\sigma}
=
\mathcal{I}_{\phi_3\sigma}
=
\mathcal{I}_{\theta\sigma}
=
\mathcal{I}_{k\sigma}
=0,$
and the remaining entries of the Fisher information matrix are the same as those in the case \(\xi \ne 0\), excluding the entries involving \(\xi\).
\begin{theorem}[Asymptotic normality of the MLE for $\xi = 0$]
Let \( X_1, X_2, \dots, X_n \) be i.i.d. observations from the EEVIMM with known \( d_0>0 \), \( u_0>d_0 \), and parameter vector
\[
\boldsymbol{\Theta}^{*} = (\phi_1, \phi_2, \theta, k, \phi_3,  \sigma)^\top
\in (0,1)^2 \times (0,\infty)^2 \times (0,1) \times (0,\infty),
\]
with \(0<\phi_1,\phi_2,\phi_3<1\) and \(\phi_1+\phi_2+\phi_3<1\). Suppose that the true parameter vector \( \boldsymbol{\Theta}^{*}_0 \) belongs to the interior of the parameter space and that the usual regularity conditions for maximum likelihood estimation are satisfied.
Let \( \widehat{\boldsymbol{\Theta}^{*}}_n \) be the MLE of $\boldsymbol{\Theta}^{*}_0$ based on the sample of size \( n \). Then, as \( n \to \infty \),
\[
\sqrt{n} \left( \widehat{\boldsymbol{\Theta}^{*}}_n - \boldsymbol{\Theta}^{*}_0 \right) 
\xrightarrow{d} \mathcal{N}_6 \left( \bm{0}, \, \mathcal{I}^{-1}(\boldsymbol{\Theta}^{*}_0) \right),
\]
where $\mathcal{I}(\boldsymbol{\Theta}^{*}_0)$ denotes the Fisher information matrix corresponding to a single observation, evaluated at the true parameter value \( \boldsymbol{\Theta}^{*}_0 \), and having entries $\mathcal{I}_{ij}(\boldsymbol{\Theta}^{*}_0)$.
The Fisher information matrix is defined similarly as in the case \(\xi \neq 0\), with the row and column corresponding to \(\xi\) removed.
\end{theorem}

\begin{remark}[Asymptotic confidence intervals for the MLE]
Based on the asymptotic normality of the MLE, the approximate \(100(1-\alpha)\%\) confidence interval for the \(i\)-th parameter \(\Theta_i^{*}\in\boldsymbol{\Theta}^{*}\) is given by 
\[
\widehat{\Theta}_i^{*}
\pm
z_{\alpha/2}
\sqrt{
\frac{
\left[
\mathcal{I}^{-1}(\widehat{\boldsymbol{\Theta}}^{*}_n)
\right]_{ii}
}{n}
},
\]
where \(z_{\alpha/2}\) denotes the upper \(\alpha/2\)-quantile of the standard normal distribution, and
\(
\left[
\mathcal{I}^{-1}(\widehat{\boldsymbol{\Theta}}^{*}_n)
\right]_{ii}
\)
denotes the \(i\)-th diagonal element of the inverse Fisher information matrix evaluated at the MLE \(\widehat{\boldsymbol{\Theta}}^{*}_n\). In the case \(\xi=0\), the confidence intervals are defined similarly after excluding the parameter \(\xi\).
\end{remark}

\begin{remark}
The mixture proportions $\phi_1$, $\phi_2$, and $\phi_3$ exhibit statistical dependence in the Fisher information matrix, as reflected by the nonzero off-diagonal terms $\mathcal{I}_{\phi_i\phi_j} \neq 0$, for $i \neq j$. This dependence arises from the unit-sum constraint
$\phi_1 + \phi_2 + \phi_3 + (1-\phi_1-\phi_2-\phi_3)=1$,
which induces asymptotic correlation among the corresponding MLEs. Consequently, changes in one proportion necessarily alter the probability mass available to the remaining components. The proposed EEVIMM framework explicitly accounts for the proportions of instantaneous failures ($\phi_1$), early failures ($\phi_2$), and extremes ($\phi_3$), whereas the existing EVMM literature has often either ignored this component or not accounted for it thoroughly. Ignoring such components may affect threshold estimation, tail inference, and consequently, the associated risk measures.
\end{remark}
\section{Monte Carlo simulation study}\label{SStudy} \label{simulation}
\noindent This section presents numerical illustrations of the EEVIMM framework. We conducted an extensive Monte Carlo simulation study using data generated from the proposed model to evaluate the performance of the estimation procedure. Then, we apply the methodology to a real-world dataset to demonstrate its practical utility.
 To study the parameter-estimation performance across various parameter settings, we generate data \(x_i\), for \(i = 1,2,3,\dots,n\), from the EEVIMM~\eqref{evimmpdf} using the simulation algorithm given in Algorithm~\ref{alg:EEVIMM_general} of Section~\ref{algorithm}. The parameters considered in the study are as \( \phi_1 \in \{0.1, 0.2, 0.3, 0.4\} \), \( \phi_2 \in \{0.1, 0.2, 0.3\} \), \( d \in \{0.5, 1\} \), \( k \in \{0.5, 1.5\}\), \( \phi_3 \in \{0.10,0.15\} \), \( u \in \{7.7263,8.4879, 9.6579, 11.8831,  13.8155 \} \), \( \xi \in \{-0.2, 0, 0.2\} \), \(\sigma = 5 \). The sample sizes considered are \( n \in \{300, 400, 500, 750, 1000,1200\} \).

The simulated data are analysed using the EEVIMM, assuming that the parameter vector
$\boldsymbol{\Theta}
= (\phi_1, \phi_2, \theta, d, k, \phi_3, u, \xi, \sigma)$ is unknown. The parameter vector \( \bm{\Theta} \) is estimated using the methodologies described in Section~\ref{estimation}. For each sample size $n$, $N=2000$ datasets are simulated, and the proposed estimation approach is applied to each dataset.
We computed the sample mean, bias, and root-mean-square error (RMSE). To assess estimation uncertainty, bootstrap standard errors (BSEs) and 95\% bootstrap confidence intervals (BCIs) are constructed using bootstrap resamples for each parameter set. Estimate the coverage probability (CI), where CI is the percentage of times the confidence interval includes the true parameter value. The bias, RMSE, and CI are calculated as
\begin{equation}
\mathrm{Bias}(\hat{\Theta}_j)
= \frac{1}{N} \sum_{i=1}^{N} \bigl(\hat{\Theta}_j^{(i)} - \Theta_j\bigr),
\quad
\mathrm{RMSE}(\hat{\Theta}_j)
= \sqrt{ \frac{1}{N} \sum_{i=1}^{N}
\bigl(\hat{\Theta}_j^{(i)} - \Theta_j\bigr)^2 },
\quad
\mathrm{CP}(\hat{\Theta}_j)
= \frac{1}{N} \sum_{i=1}^{N} \mathbf{1}\{{\Theta}_0 \in CI_i\},
\end{equation}
\noindent where $\hat{\Theta}_j$ denotes the estimator of the true parameter $\Theta_j$, and $\hat{\Theta}_j^{(i)}$ denotes the estimate of $\Theta_j$ obtained
from the $i$th simulated dataset. Here, \( N \) is the total number of simulated datasets, each generated with sample size \( n \). $\bm{\Theta_0}$ is the true parameter value, and $CI_i$ is the confidence interval obtained in the $i$-th replication.
 The sample mean, BCI, BSE, MSE, and bias using estimated parameter values with their respective sample sizes are given in Tables \ref{Table-100} - \ref{Table-103}, respectively. These tables report the results for the heavy-tailed, lighter-tailed, and exponential-tailed scenarios, respectively. Additional results are given in the Appendix. We see that the estimated parameter values by the proposed methodologies are close to the true parameter values. The parameters $u$, $\xi$ and $\sigma$ exhibit some variability at smaller sample sizes. As the sample size increases, both bias and MSE decrease, indicating improved estimation accuracy and greater estimator stability.
 The bootstrap BCIs provide expected results for estimating model parameters using the proposed methodologies.  
 
The simulated data are generated from the proposed EEVIMM under different tail scenarios. The threshold ($u$) and extreme-value parameters ($\xi, \sigma$) are estimated using different classical threshold selection methods, including numerical fixed-quantile methods, graphical methods, EVMM-based methods, and other approaches.  In the numerical fixed-quantile methods, the $90$th percentile, $k=\sqrt{n}$, and $k=n^{2/3}/\log\log n$ are considered. Graphical threshold selection methods, including MEP, PSP, PP, and HP, are considered. Some other estimation methods,
such as the DPT method (\cite{sakthivel2025dual}), are considered. In the mixture model framework, EVMM (\cite{behrens2004bayesian}), FEVMM (\cite{hu2013extreme}),
M\_EVMM (\cite{do2012semiparametric}) and the recently proposed model, EVIMM (\cite{nila2025modeling}), are used to estimate the threshold and extreme-value parameters, which are then compared, and the results are shown in Tables~\ref{Table-106} - \ref{Table-108}.
Values in parentheses in the \(\phi_u\) column denote empirical estimates of the tail proportion, corresponding to \(\phi_3\) in the proposed EEVIMM. The number of exceedances above the threshold \(u\), reported in the \(N_u\) column, is given by \(N_u = n\phi_3\).
The \texttt{evmix} R package developed by~\cite{hu2018evmix}, which is available on CRAN, is used for this purpose. Estimates of the parameters $u$, $\xi$, and $\sigma$ vary across cases and may deviate from their true values, indicating unreliable tail estimation. In contrast, the EEVIMM estimates remain closer to the true values and exhibit more accurate tail estimation in the presence of inliers.
\begin{landscape}
\begin{table*}[htbp]
\centering
\scriptsize
\setlength{\tabcolsep}{3pt}
\renewcommand{\arraystretch}{0.9}
 \captionsetup{
      labelfont=bf,
      textfont=normal
    }
 \caption{Monte Carlo simulation results ($N=2000$) for the heavy-tailed scenario with sample sizes $n=300,400,500,750,1000,$ and $1200$, reporting Sample mean, BSE, 95\% BCI, RMSE, Bias, and CP.}
\begin{tabular}{|p{2.5cm}|p{1.8cm}|p{1.8cm}|p{2.3cm}|p{1.8cm}|p{2.6cm}|p{2.3cm}|p{2.3cm}|p{1.5cm}|}
\hline
\textbf{$n$} & \textbf{Parameters} & \textbf{True Value} &
\textbf{Sample Mean} & \textbf{BSE} & \textbf{BCI} &
\textbf{RMSE} & \textbf{Bias} & \textbf{CP} \\ \hline
\multirow{9}{*}{ 300 }
 & $\phi_1$ &  0.3  &  0.3001  &  0.0250  & ( 0.2000 ,  0.2967 )
 &  0.0263  &  0.0001  &  91.7  \\ \cline{2-9}
 & $\phi_2$ &  0.2  &  0.1812  &  0.0590  & ( 0.0667 ,  0.2733 )
 &  0.0553  &  -0.0189  &  100.0  \\ \cline{2-9}
 & $\theta$ &  2  &  2.1184  &  0.3884  & ( 1.4628 ,  2.8827 )
 &  0.4659  &  0.1184  &  95.0  \\ \cline{2-9}
 & $d$ &  1  &  0.9407  &  0.1858  & ( 0.5211 ,  1.2201 )
 &  0.1811  &  -0.0593  &  96.7  \\ \cline{2-9}
 & $k$ &  0.5  &  0.4952  &  0.1174  & ( 0.3000 ,  0.7178 )
 &  0.1224  &  -0.0048  &  100.0  \\ \cline{2-9}
 & $\phi_3$ &  0.1  &  0.1225  &  0.0190  & ( 0.1000 ,  0.1600 )
 &  0.0293  &  0.0225  &  93.3  \\ \cline{2-9}
 & $u$ &  9.6579  &  8.1125  &  0.8774  & ( 5.2539 ,  9.1437 )
 &  2.4016  &  -1.5454  &  91.7  \\ \cline{2-9}
 & $\xi$ &  0.2  &  0.0820  &  0.3028  & ( -0.0972 ,  1.1582 )
 &  0.3072  &  -0.1180  &  93.3  \\ \cline{2-9}
 & $\sigma$ &  5  &  5.9601  &  1.3599  & ( 1.7002 ,  7.2260 )
 &  2.4977  &  0.9601  &  93.3  \\ \hline

\multirow{9}{*}{ 400 }
 & $\phi_1$ &  0.3  &  0.3004  &  0.0228  & ( 0.2375 ,  0.3225 )
 &  0.0230  &  0.0004  &  91.7  \\ \cline{2-9}
 & $\phi_2$ &  0.2  &  0.1808  &  0.0507  & ( 0.0625 ,  0.2525 )
 &  0.0536  &  -0.0192  &  100.0  \\ \cline{2-9}
 & $\theta$ &  2  &  2.0876  &  0.3462  & ( 1.5851 ,  2.9040 )
 &  0.3856  &  0.0876  &  98.3  \\ \cline{2-9}
 & $d$ &  1  &  0.9425  &  0.1586  & ( 0.5472 ,  1.1705 )
 &  0.1718  &  -0.0575  &  98.3  \\ \cline{2-9}
 & $k$ &  0.5  &  0.5030  &  0.0860  & ( 0.3000 ,  0.6642 )
 &  0.1012  &  0.0030  &  100.0  \\ \cline{2-9}
 & $\phi_3$ &  0.1  &  0.1109  &  0.0205  & ( 0.0750 ,  0.1575 )
 &  0.0251  &  0.0109  &  98.3  \\ \cline{2-9}
 & $u$ &  9.6579  &  9.0336  &  1.3597  & ( 5.4786 ,  11.2215 )
 &  1.8682  &  -0.6243  &  95.0  \\ \cline{2-9}
 & $\xi$ &  0.2  &  0.1625  &  0.2462  & ( -0.4802 ,  0.4948 )
 &  0.2665  &  -0.0375  &  91.7  \\ \cline{2-9}
 & $\sigma$ &  5  &  5.2925  &  1.6503  & ( 2.1440 ,  9.0474 )
 &  2.0211  &  0.2925  &  96.7  \\ \hline

\multirow{9}{*}{ 500 }
 & $\phi_1$ &  0.3  &  0.3003  &  0.0205  & ( 0.2400 ,  0.3220 )
 &  0.0206  &  0.0003  &  93.3  \\ \cline{2-9}
 & $\phi_2$ &  0.2  &  0.1796  &  0.0541  & ( 0.0400 ,  0.2460 )
 &  0.0528  &  -0.0204  &  98.3  \\ \cline{2-9}
 & $\theta$ &  2  &  2.0651  &  0.3918  & ( 1.3796 ,  2.9350 )
 &  0.3236  &  0.0651  &  91.7  \\ \cline{2-9}
 & $d$ &  1  &  0.9393  &  0.1799  & ( 0.4316 ,  1.1521 )
 &  0.1665  &  -0.0607  &  95.0  \\ \cline{2-9}
 & $k$ &  0.5  &  0.5071  &  0.0653  & ( 0.3000 ,  0.5608 )
 &  0.0896  &  0.0071  &  100.0  \\ \cline{2-9}
 & $\phi_3$ &  0.1  &  0.1082  &  0.0177  & ( 0.0700 ,  0.1480 )
 &  0.0228  &  0.0082  &  100.0  \\ \cline{2-9}
 & $u$ &  9.6579  &  9.2267  &  1.5415  & ( 6.1560 ,  12.2812 )
 &  1.6061  &  -0.4312  &  100.0  \\ \cline{2-9}
 & $\xi$ &  0.2  &  0.1698  &  0.3326  & ( -0.0652 ,  1.2427 )
 &  0.2276  &  -0.0302  &  95.0  \\ \cline{2-9}
 & $\sigma$ &  5  &  5.1922  &  1.1958  & ( 1.4431 ,  6.2025 )
 &  1.7203  &  0.1922  &  98.3  \\ \hline

\multirow{9}{*}{ 750 }
 & $\phi_1$ &  0.3  &  0.2997  &  0.0161  & ( 0.2360 ,  0.3000 )
 &  0.0171  &  -0.0003  &  93.3  \\ \cline{2-9}
 & $\phi_2$ &  0.2  &  0.1780  &  0.0508  & ( 0.0680 ,  0.2547 )
 &  0.0515  &  -0.0220  &  95.0  \\ \cline{2-9}
 & $\theta$ &  2  &  2.0369  &  0.2212  & ( 1.5980 ,  2.4945 )
 &  0.2447  &  0.0369  &  95.0  \\ \cline{2-9}
 & $d$ &  1  &  0.9341  &  0.1588  & ( 0.5704 ,  1.1631 )
 &  0.1592  &  -0.0659  &  95.0  \\ \cline{2-9}
 & $k$ &  0.5  &  0.5109  &  0.0920  & ( 0.3000 ,  0.6708 )
 &  0.0687  &  0.0109  &  96.7  \\ \cline{2-9}
 & $\phi_3$ &  0.1  &  0.1053  &  0.0270  & ( 0.0707 ,  0.1587 )
 &  0.0190  &  0.0053  &  98.3  \\ \cline{2-9}
 & $u$ &  9.6579  &  9.4151  &  2.3122  & ( 4.7242 ,  11.9646 )
 &  1.2669  &  -0.2428  &  98.3  \\ \cline{2-9}
 & $\xi$ &  0.2  &  0.1843  &  0.1740  & ( -0.3576 ,  0.3274 )
 &  0.1719  &  -0.0157  &  91.7  \\ \cline{2-9}
 & $\sigma$ &  5  &  5.0550  &  1.8419  & ( 4.2905 ,  11.2615 )
 &  1.1948  &  0.0550  &  91.7  \\ \hline

\multirow{9}{*}{ 1000 }
 & $\phi_1$ &  0.3  &  0.2999  &  0.0149  & ( 0.2810 ,  0.3390 )
 &  0.0147  &  -0.0001  &  98.3  \\ \cline{2-9}
 & $\phi_2$ &  0.2  &  0.1765  &  0.0427  & ( 0.0610 ,  0.2170 )
 &  0.0519  &  -0.0235  &  91.7  \\ \cline{2-9}
 & $\theta$ &  2  &  2.0301  &  0.2366  & ( 1.8470 ,  2.7731 )
 &  0.2213  &  0.0301  &  96.7  \\ \cline{2-9}
 & $d$ &  1  &  0.9297  &  0.1276  & ( 0.5756 ,  1.0483 )
 &  0.1600  &  -0.0703  &  93.3  \\ \cline{2-9}
 & $k$ &  0.5  &  0.5113  &  0.0503  & ( 0.3714 ,  0.5737 )
 &  0.0562  &  0.0113  &  96.7  \\ \cline{2-9}
 & $\phi_3$ &  0.1  &  0.1036  &  0.0115  & ( 0.0880 ,  0.1320 )
 &  0.0162  &  0.0036  &  96.7  \\ \cline{2-9}
 & $u$ &  9.6579  &  9.5035  &  0.3906  & ( 8.9139 ,  10.4212 )
 &  1.0344  &  -0.1544  &  95.0  \\ \cline{2-9}
 & $\xi$ &  0.2  &  0.1895  &  0.1395  & ( 0.0108 ,  0.5631 )
 &  0.1438  &  -0.0105  &  91.7  \\ \cline{2-9}
 & $\sigma$ &  5  &  5.0227  &  0.6362  & ( 2.6049 ,  5.1326 )
 &  0.9734  &  0.0227  &  96.7  \\ \hline

\multirow{9}{*}{ 1200 }
 & $\phi_1$ &  0.3  &  0.3001  &  0.0129  & ( 0.2642 ,  0.3150 )
 &  0.0134  &  0.0001  &  93.3  \\ \cline{2-9}
 & $\phi_2$ &  0.2  &  0.1764  &  0.0420  & ( 0.0633 ,  0.2192 )
 &  0.0504  &  -0.0236  &  93.3  \\ \cline{2-9}
 & $\theta$ &  2  &  2.0279  &  0.2220  & ( 1.8982 ,  2.7895 )
 &  0.1938  &  0.0279  &  98.3  \\ \cline{2-9}
 & $d$ &  1  &  0.9300  &  0.1146  & ( 0.5368 ,  0.9687 )
 &  0.1551  &  -0.0700  &  91.7  \\ \cline{2-9}
 & $k$ &  0.5  &  0.5104  &  0.0494  & ( 0.4420 ,  0.6307 )
 &  0.0493  &  0.0104  &  91.7  \\ \cline{2-9}
 & $\phi_3$ &  0.1  &  0.1023  &  0.0159  & ( 0.0750 ,  0.1492 )
 &  0.0142  &  0.0023  &  96.7  \\ \cline{2-9}
 & $u$ &  9.6579  &  9.5831  &  1.1616  & ( 5.8166 ,  11.3874 )
 &  0.8554  &  -0.0748  &  96.7  \\ \cline{2-9}
 & $\xi$ &  0.2  &  0.1963  &  0.1325  & ( -0.0442 ,  0.4819 )
 &  0.1291  &  -0.0037  &  95.0  \\ \cline{2-9}
 & $\sigma$ &  5  &  4.9858  &  1.0200  & ( 4.1214 ,  8.2986 )
 &  0.8558  &  -0.0142  &  95.0  \\ \hline
\end{tabular}
    \label{Table-100}
\end{table*}
\end{landscape}

\begin{landscape}
\begin{table*}[htbp]
\centering
\scriptsize
\setlength{\tabcolsep}{3pt}
\renewcommand{\arraystretch}{0.9}
 \captionsetup{
      labelfont=bf,
      textfont=normal
    }
   \caption{Monte Carlo simulation results ($N=2000$) for the light-tailed scenario with sample sizes $n=300,400,500,750,1000,$ and $1200$, reporting Sample mean, BSE, 95\% BCI, RMSE, Bias, and CP.}
\begin{tabular}{|p{2.5cm}|p{1.8cm}|p{1.8cm}|p{2.3cm}|p{1.8cm}|p{2.6cm}|p{2.3cm}|p{2.3cm}|p{1.5cm}|}
\hline
\textbf{$n$} & \textbf{Parameters} & \textbf{True Value} &
\textbf{Sample Mean} & \textbf{BSE} & \textbf{BCI} &
\textbf{RMSE} & \textbf{Bias} & \textbf{CP} \\ \hline
\multirow{9}{*}{ 300 }
 & $\phi_1$ &  0.3  &  0.3001  &  0.0265  & ( 0.2633 ,  0.3667 )
 &  0.0263  &  0.0001  &  95.0  \\ \cline{2-9}
 & $\phi_2$ &  0.2  &  0.1812  &  0.0641  & ( 0.0500 ,  0.2533 )
 &  0.0553  &  -0.0189  &  98.3  \\ \cline{2-9}
 & $\theta$ &  2  &  2.1184  &  0.3385  & ( 0.9409 ,  2.2750 )
 &  0.4659  &  0.1184  &  95.0  \\ \cline{2-9}
 & $d$ &  1  &  0.9407  &  0.3116  & ( 0.5414 ,  1.6693 )
 &  0.1811  &  -0.0593  &  98.3  \\ \cline{2-9}
 & $k$ &  0.5  &  0.4860  &  0.0992  & ( 0.3000 ,  0.6592 )
 &  0.1244  &  -0.0140  &  100.0  \\ \cline{2-9}
 & $\phi_3$ &  0.1  &  0.1228  &  0.0179  & ( 0.1000 ,  0.1600 )
 &  0.0295  &  0.0228  &  93.3  \\ \cline{2-9}
 & $u$ &  9.6579  &  8.0221  &  1.1592  & ( 6.2515 ,  11.4454 )
 &  2.4657  &  -1.6358  &  88.3  \\ \cline{2-9}
 & $\xi$ &  -0.2  &  -0.3435  &  0.2129  & ( -0.7252 ,  0.1489 )
 &  0.2863  &  -0.1435  &  93.3  \\ \cline{2-9}
 & $\sigma$ &  5  &  6.6266  &  1.8583  & ( 2.8574 ,  10.5628 )
 &  3.0175  &  1.6266  &  85.0  \\ \hline

\multirow{9}{*}{ 400 }
 & $\phi_1$ &  0.3  &  0.3004  &  0.0228  & ( 0.2525 ,  0.3400 )
 &  0.0230  &  0.0004  &  95.0  \\ \cline{2-9}
 & $\phi_2$ &  0.2  &  0.1808  &  0.0568  & ( 0.0650 ,  0.2575 )
 &  0.0536  &  -0.0192  &  98.3  \\ \cline{2-9}
 & $\theta$ &  2  &  2.0876  &  0.3429  & ( 1.5922 ,  2.9518 )
 &  0.3856  &  0.0876  &  96.7  \\ \cline{2-9}
 & $d$ &  1  &  0.9425  &  0.1630  & ( 0.5288 ,  1.1407 )
 &  0.1718  &  -0.0575  &  96.7  \\ \cline{2-9}
 & $k$ &  0.5  &  0.4960  &  0.1118  & ( 0.3000 ,  0.7657 )
 &  0.1038  &  -0.0040  &  100.0  \\ \cline{2-9}
 & $\phi_3$ &  0.1  &  0.1118  &  0.0227  & ( 0.0750 ,  0.1575 )
 &  0.0255  &  0.0118  &  100.0  \\ \cline{2-9}
 & $u$ &  9.6579  &  8.9146  &  2.1223  & ( 4.4117 ,  11.2205 )
 &  1.9240  &  -0.7433  &  100.0  \\ \cline{2-9}
 & $\xi$ &  -0.2  &  -0.2676  &  0.2526  & ( -0.8859 ,  0.1472 )
 &  0.2344  &  -0.0676  &  93.3  \\ \cline{2-9}
 & $\sigma$ &  5  &  5.7185  &  2.2563  & ( 2.0622 ,  10.1917 )
 &  2.3085  &  0.7185  &  96.7  \\ \hline

\multirow{9}{*}{ 500 }
 & $\phi_1$ &  0.3  &  0.3003  &  0.0199  & ( 0.2740 ,  0.3540 )
 &  0.0206  &  0.0003  &  90.0  \\ \cline{2-9}
 & $\phi_2$ &  0.2  &  0.1796  &  0.0521  & ( 0.0640 ,  0.2460 )
 &  0.0528  &  -0.0204  &  98.3  \\ \cline{2-9}
 & $\theta$ &  2  &  2.0651  &  0.2936  & ( 1.6048 ,  2.7878 )
 &  0.3236  &  0.0651  &  96.7  \\ \cline{2-9}
 & $d$ &  1  &  0.9393  &  0.1570  & ( 0.5344 ,  1.1071 )
 &  0.1665  &  -0.0607  &  95.0  \\ \cline{2-9}
 & $k$ &  0.5  &  0.5032  &  0.0910  & ( 0.3000 ,  0.6652 )
 &  0.0921  &  0.0032  &  98.3  \\ \cline{2-9}
 & $\phi_3$ &  0.1  &  0.1092  &  0.0229  & ( 0.0700 ,  0.1580 )
 &  0.0233  &  0.0092  &  96.7  \\ \cline{2-9}
 & $u$ &  9.6579  &  9.1238  &  2.0311  & ( 5.2568 ,  12.1659 )
 &  1.6535  &  -0.5341  &  98.3  \\ \cline{2-9}
 & $\xi$ &  -0.2  &  -0.2530  &  0.2365  & ( -0.6050 ,  0.3158 )
 &  0.1964  &  -0.0530  &  93.3  \\ \cline{2-9}
 & $\sigma$ &  5  &  5.5020  &  1.8756  & ( 2.0824 ,  9.2663 )
 &  1.9068  &  0.5020  &  95.0  \\ \hline

\multirow{9}{*}{ 750 }
 & $\phi_1$ &  0.3  &  0.2997  &  0.0171  & ( 0.2947 ,  0.3613 )
 &  0.0171  &  -0.0003  &  91.7  \\ \cline{2-9}
 & $\phi_2$ &  0.2  &  0.1780  &  0.0433  & ( 0.0613 ,  0.2187 )
 &  0.0515  &  -0.0220  &  91.7  \\ \cline{2-9}
 & $\theta$ &  2  &  2.0369  &  0.2504  & ( 1.7065 ,  2.7200 )
 &  0.2447  &  0.0369  &  95.0  \\ \cline{2-9}
 & $d$ &  1  &  0.9341  &  0.1457  & ( 0.5891 ,  1.1361 )
 &  0.1592  &  -0.0659  &  90.0  \\ \cline{2-9}
 & $k$ &  0.5  &  0.5079  &  0.0629  & ( 0.3289 ,  0.5913 )
 &  0.0706  &  0.0079  &  98.3  \\ \cline{2-9}
 & $\phi_3$ &  0.1  &  0.1062  &  0.0146  & ( 0.0827 ,  0.1427 )
 &  0.0196  &  0.0062  &  98.3  \\ \cline{2-9}
 & $u$ &  9.6579  &  9.3323  &  0.6646  & ( 7.9056 ,  10.8448 )
 &  1.3482  &  -0.3256  &  98.3  \\ \cline{2-9}
 & $\xi$ &  -0.2  &  -0.2332  &  0.1277  & ( -0.3366 ,  0.1593 )
 &  0.1412  &  -0.0332  &  93.3  \\ \cline{2-9}
 & $\sigma$ &  5  &  5.2724  &  0.8602  & ( 2.8304 ,  6.1354 )
 &  1.3726  &  0.2724  &  91.7  \\ \hline

\multirow{9}{*}{ 1000 }
 & $\phi_1$ &  0.3  &  0.2999  &  0.0148  & ( 0.2900 ,  0.3480 )
 &  0.0147  &  -0.0001  &  96.7  \\ \cline{2-9}
 & $\phi_2$ &  0.2  &  0.1765  &  0.0452  & ( 0.0650 ,  0.2350 )
 &  0.0519  &  -0.0235  &  96.7  \\ \cline{2-9}
 & $\theta$ &  2  &  2.0301  &  0.1808  & ( 1.5058 ,  2.2193 )
 &  0.2213  &  0.0301  &  98.3  \\ \cline{2-9}
 & $d$ &  1  &  0.9297  &  0.1566  & ( 0.5436 ,  1.1246 )
 &  0.1600  &  -0.0703  &  93.3  \\ \cline{2-9}
 & $k$ &  0.5  &  0.5101  &  0.0798  & ( 0.3000 ,  0.5740 )
 &  0.0583  &  0.0101  &  93.3  \\ \cline{2-9}
 & $\phi_3$ &  0.1  &  0.1043  &  0.0249  & ( 0.0710 ,  0.1590 )
 &  0.0169  &  0.0043  &  95.0  \\ \cline{2-9}
 & $u$ &  9.6579  &  9.4418  &  1.6970  & ( 5.4560 ,  11.3447 )
 &  1.0984  &  -0.2161  &  93.3  \\ \cline{2-9}
 & $\xi$ &  -0.2  &  -0.2241  &  0.1186  & ( -0.5755 ,  -0.1165 )
 &  0.1155  &  -0.0241  &  91.7  \\ \cline{2-9}
 & $\sigma$ &  5  &  5.1754  &  1.4135  & ( 3.7868 ,  8.8529 )
 &  1.0850  &  0.1754  &  93.3  \\ \hline

\multirow{9}{*}{ 1200 }
 & $\phi_1$ &  0.3  &  0.3001  &  0.0138  & ( 0.2850 ,  0.3392 )
 &  0.0134  &  0.0001  &  96.7  \\ \cline{2-9}
 & $\phi_2$ &  0.2  &  0.1764  &  0.0494  & ( 0.0700 ,  0.2425 )
 &  0.0504  &  -0.0236  &  91.7  \\ \cline{2-9}
 & $\theta$ &  2  &  2.0279  &  0.1491  & ( 1.4127 ,  2.0025 )
 &  0.1938  &  0.0279  &  98.3  \\ \cline{2-9}
 & $d$ &  1  &  0.9300  &  0.1743  & ( 0.5451 ,  1.2050 )
 &  0.1551  &  -0.0700  &  91.7  \\ \cline{2-9}
 & $k$ &  0.5  &  0.5102  &  0.0847  & ( 0.4355 ,  0.7657 )
 &  0.0502  &  0.0102  &  93.3  \\ \cline{2-9}
 & $\phi_3$ &  0.1  &  0.1028  &  0.0216  & ( 0.0817 ,  0.1600 )
 &  0.0144  &  0.0028  &  98.3  \\ \cline{2-9}
 & $u$ &  9.6579  &  9.5370  &  1.5282  & ( 4.1977 ,  9.6035 )
 &  0.8821  &  -0.1209  &  93.3  \\ \cline{2-9}
 & $\xi$ &  -0.2  &  -0.2169  &  0.0787  & ( -0.5447 ,  -0.2322 )
 &  0.1008  &  -0.0169  &  95.0  \\ \cline{2-9}
 & $\sigma$ &  5  &  5.1017  &  1.2105  & ( 5.8588 ,  10.4552 )
 &  0.9105  &  0.1017  &  96.7  \\ \hline

 \end{tabular}
    \label{Table-101}
\end{table*}
\end{landscape}

\begin{landscape}
\begin{table*}[htbp]
\centering
\scriptsize
\setlength{\tabcolsep}{3pt}
\renewcommand{\arraystretch}{0.9}
 \captionsetup{
      labelfont=bf,
      textfont=normal
    }
   \caption{Monte Carlo simulation results ($N=2000$) for the exponential-tailed scenario with sample sizes $n=300,400,500,750,1000,$ and $1200$, reporting Sample mean, BSE, 95\% BCI, RMSE, Bias, and CP.}
\begin{tabular}{|p{2.5cm}|p{1.8cm}|p{1.8cm}|p{2.3cm}|p{1.8cm}|p{2.6cm}|p{2.3cm}|p{2.3cm}|p{1.5cm}|}
\hline
\textbf{$n$} & \textbf{Parameters} & \textbf{True Value} &
\textbf{Sample Mean} & \textbf{BSE} & \textbf{BCI} &
\textbf{RMSE} & \textbf{Bias} & \textbf{CP} \\ \hline
\multirow{9}{*}{ 300 }
 & $\phi_1$ &  0.3  &  0.3001  &  0.0268  & ( 0.2433 ,  0.3500 )
 &  0.0263  &  0.0001  &  95.0  \\ \cline{2-9}
 & $\phi_2$ &  0.2  &  0.1812  &  0.0538  & ( 0.0633 ,  0.2567 )
 &  0.0553  &  -0.0189  &  100.0  \\ \cline{2-9}
 & $\theta$ &  2  &  2.1184  &  0.3919  & ( 1.3581 ,  2.8252 )
 &  0.4659  &  0.1184  &  95.0  \\ \cline{2-9}
 & $d$ &  1  &  0.9407  &  0.2064  & ( 0.5192 ,  1.3285 )
 &  0.1811  &  -0.0593  &  91.7  \\ \cline{2-9}
 & $k$ &  0.5  &  0.4900  &  0.1487  & ( 0.3000 ,  0.7972 )
 &  0.1235  &  -0.0100  &  100.0  \\ \cline{2-9}
 & $\phi_3$ &  0.1  &  0.1228  &  0.0193  & ( 0.1000 ,  0.1600 )
 &  0.0296  &  0.0228  &  95.0  \\ \cline{2-9}
 & $u$ &  9.6579  &  8.0561  &  1.2199  & ( 4.2721 ,  9.0433 )
 &  2.4531  &  -1.6018  &  71.7  \\ \cline{2-9}
 & $\xi$ &  0  &  -0.1291  &  0.2791  & ( -0.3724 ,  0.7656 )
 &  0.2932  &  -0.1291  &  80.0  \\ \cline{2-9}
 & $\sigma$ &  5  &  6.2756  &  1.9718  & ( 2.4786 ,  10.0008 )
 &  2.7358  &  1.2756  &  86.7  \\ \hline

\multirow{9}{*}{ 400 }
 & $\phi_1$ &  0.3  &  0.3004  &  0.0239  & ( 0.2375 ,  0.3325 )
 &  0.0230  &  0.0004  &  98.3  \\ \cline{2-9}
 & $\phi_2$ &  0.2  &  0.1808  &  0.0473  & ( 0.0550 ,  0.2375 )
 &  0.0536  &  -0.0192  &  98.3  \\ \cline{2-9}
 & $\theta$ &  2  &  2.0876  &  0.5006  & ( 1.8378 ,  3.8512 )
 &  0.3856  &  0.0876  &  96.7  \\ \cline{2-9}
 & $d$ &  1  &  0.9425  &  0.1428  & ( 0.5495 ,  1.1497 )
 &  0.1718  &  -0.0575  &  95.0  \\ \cline{2-9}
 & $k$ &  0.5  &  0.4997  &  0.0929  & ( 0.3000 ,  0.7469 )
 &  0.1024  &  -0.0003  &  100.0  \\ \cline{2-9}
 & $\phi_3$ &  0.1  &  0.1112  &  0.0196  & ( 0.0750 ,  0.1525 )
 &  0.0253  &  0.0112  &  100.0  \\ \cline{2-9}
 & $u$ &  9.6579  &  8.9889  &  1.5795  & ( 4.5423 ,  10.5345 )
 &  1.8952  &  -0.6690  &  96.7  \\ \cline{2-9}
 & $\xi$ &  0  &  -0.0504  &  0.2129  & ( -0.7893 ,  0.0665 )
 &  0.2475  &  -0.0504  &  90.0  \\ \cline{2-9}
 & $\sigma$ &  5  &  5.4747  &  2.1034  & ( 2.7428 ,  10.9283 )
 &  2.1211  &  0.4747  &  95.0  \\ \hline

\multirow{9}{*}{ 500 }
 & $\phi_1$ &  0.3  &  0.3003  &  0.0213  & ( 0.2780 ,  0.3620 )
 &  0.0206  &  0.0003  &  98.3  \\ \cline{2-9}
 & $\phi_2$ &  0.2  &  0.1796  &  0.0473  & ( 0.0620 ,  0.2360 )
 &  0.0528  &  -0.0204  &  100.0  \\ \cline{2-9}
 & $\theta$ &  2  &  2.0651  &  0.3325  & ( 1.6757 ,  2.9815 )
 &  0.3236  &  0.0651  &  100.0  \\ \cline{2-9}
 & $d$ &  1  &  0.9393  &  0.1510  & ( 0.6049 ,  1.1834 )
 &  0.1665  &  -0.0607  &  96.7  \\ \cline{2-9}
 & $k$ &  0.5  &  0.5055  &  0.0923  & ( 0.3000 ,  0.6932 )
 &  0.0908  &  0.0055  &  98.3  \\ \cline{2-9}
 & $\phi_3$ &  0.1  &  0.1085  &  0.0211  & ( 0.0740 ,  0.1580 )
 &  0.0229  &  0.0085  &  100.0  \\ \cline{2-9}
 & $u$ &  9.6579  &  9.1921  &  1.6603  & ( 4.7413 ,  11.1435 )
 &  1.6099  &  -0.4658  &  98.3  \\ \cline{2-9}
 & $\xi$ &  0  &  -0.0400  &  0.1913  & ( -0.5410 ,  0.2120 )
 &  0.2086  &  -0.0400  &  86.7  \\ \cline{2-9}
 & $\sigma$ &  5  &  5.3212  &  1.8114  & ( 2.9717 ,  10.0672 )
 &  1.7682  &  0.3212  &  91.7  \\ \hline

\multirow{9}{*}{ 750 }
 & $\phi_1$ &  0.3  &  0.2997  &  0.0161  & ( 0.2506 ,  0.3147 )
 &  0.0171  &  -0.0003  &  95.0  \\ \cline{2-9}
 & $\phi_2$ &  0.2  &  0.1780  &  0.0459  & ( 0.0627 ,  0.2333 )
 &  0.0515  &  -0.0220  &  91.7  \\ \cline{2-9}
 & $\theta$ &  2  &  2.0369  &  0.2616  & ( 1.7177 ,  2.7041 )
 &  0.2447  &  0.0369  &  96.7  \\ \cline{2-9}
 & $d$ &  1  &  0.9341  &  0.1366  & ( 0.5596 ,  1.0671 )
 &  0.1592  &  -0.0659  &  86.7  \\ \cline{2-9}
 & $k$ &  0.5  &  0.5097  &  0.0635  & ( 0.3000 ,  0.5152 )
 &  0.0695  &  0.0097  &  96.7  \\ \cline{2-9}
 & $\phi_3$ &  0.1  &  0.1057  &  0.0268  & ( 0.0707 ,  0.1600 )
 &  0.0192  &  0.0057  &  95.0  \\ \cline{2-9}
 & $u$ &  9.6579  &  9.3849  &  2.2819  & ( 6.7920 ,  14.3953 )
 &  1.2978  &  -0.2730  &  93.3  \\ \cline{2-9}
 & $\xi$ &  0  &  -0.0233  &  0.1555  & ( -0.5536 ,  0.0542 )
 &  0.1539  &  -0.0233  &  86.7  \\ \cline{2-9}
 & $\sigma$ &  5  &  5.1466  &  1.8024  & ( 3.5746 ,  10.2725 )
 &  1.2557  &  0.1466  &  96.7  \\ \hline

\multirow{9}{*}{ 1000 }
 & $\phi_1$ &  0.3  &  0.2999  &  0.0145  & ( 0.2680 ,  0.3250 )
 &  0.0147  &  -0.0001  &  100.0  \\ \cline{2-9}
 & $\phi_2$ &  0.2  &  0.1765  &  0.0482  & ( 0.0710 ,  0.2460 )
 &  0.0519  &  -0.0235  &  86.7  \\ \cline{2-9}
 & $\theta$ &  2  &  2.0301  &  0.1902  & ( 1.6107 ,  2.3653 )
 &  0.2213  &  0.0301  &  100.0  \\ \cline{2-9}
 & $d$ &  1  &  0.9297  &  0.1403  & ( 0.5553 ,  1.0747 )
 &  0.1600  &  -0.0703  &  88.3  \\ \cline{2-9}
 & $k$ &  0.5  &  0.5108  &  0.0966  & ( 0.3000 ,  0.6672 )
 &  0.0576  &  0.0108  &  96.7  \\ \cline{2-9}
 & $\phi_3$ &  0.1  &  0.1041  &  0.0275  & ( 0.0720 ,  0.1600 )
 &  0.0167  &  0.0041  &  96.7  \\ \cline{2-9}
 & $u$ &  9.6579  &  9.4677  &  1.4701  & ( 5.6564 ,  10.8970 )
 &  1.0796  &  -0.1902  &  98.3  \\ \cline{2-9}
 & $\xi$ &  0  &  -0.0170  &  0.1135  & ( -0.3105 ,  0.1487 )
 &  0.1273  &  -0.0170  &  93.3  \\ \cline{2-9}
 & $\sigma$ &  5  &  5.0916  &  0.9652  & ( 4.0358 ,  8.0260 )
 &  1.0082  &  0.0916  &  98.3  \\ \hline

\multirow{9}{*}{ 1200 }
 & $\phi_1$ &  0.3  &  0.3001  &  0.0134  & ( 0.2850 ,  0.3375 )
 &  0.0134  &  0.0001  &  96.7  \\ \cline{2-9}
 & $\phi_2$ &  0.2  &  0.1764  &  0.0374  & ( 0.0525 ,  0.1900 )
 &  0.0504  &  -0.0236  &  90.0  \\ \cline{2-9}
 & $\theta$ &  2  &  2.0279  &  0.2279  & ( 1.8118 ,  2.7075 )
 &  0.1938  &  0.0279  &  98.3  \\ \cline{2-9}
 & $d$ &  1  &  0.9300  &  0.1190  & ( 0.5093 ,  0.9625 )
 &  0.1551  &  -0.0700  &  93.3  \\ \cline{2-9}
 & $k$ &  0.5  &  0.5102  &  0.0605  & ( 0.3000 ,  0.5871 )
 &  0.0494  &  0.0102  &  98.3  \\ \cline{2-9}
 & $\phi_3$ &  0.1  &  0.1024  &  0.0214  & ( 0.0717 ,  0.1558 )
 &  0.0141  &  0.0024  &  98.3  \\ \cline{2-9}
 & $u$ &  9.6579  &  9.5694  &  1.6550  & ( 5.2003 ,  11.3979 )
 &  0.8562  &  -0.0885  &  95.0  \\ \cline{2-9}
 & $\xi$ &  0  &  -0.0093  &  0.1232  & ( -0.3240 ,  0.1456 )
 &  0.1131  &  -0.0093  &  93.3  \\ \cline{2-9}
 & $\sigma$ &  5  &  5.0314  &  1.2491  & ( 3.9414 ,  8.8486 )
 &  0.8643  &  0.0314  &  93.3  \\ \hline
 \end{tabular}
    \label{Table-103}
\end{table*}
\end{landscape}

\begin{table}[!t]
\centering
\small
\textbf{Threshold and tail parameter estimates from various methods}
\vspace{2mm}
\begin{tabular}{lllcccc}
\hline
\textbf{Type} & \textbf{Method / Model} & $u$ & $\xi$ & $\sigma$ & $N_u$ & $\phi_3$/$\phi_u$ \\
\hline
\multicolumn{2}{l}{\textbf{True value}} 
 & $11.8831$ & $0.20$ & $5.00$ & $150$ & $0.15$ \\
\hline
\multirow{3}{*}{Numerical methods} 
 & $90$th percentile 
&  14.5779 & 0.2650 & 5.3026 & 100 & -- \\
 & $k=\sqrt{n}$     
& 21.4138 & 0.2681 & 7.6612 &  31 & -- \\
 & $k=n^{2/3}/\log\log n$ 
& 18.4621 & 0.3669 & 5.5696 &  51 & -- \\
\hline

\multirow{3}{*}{Graphical methods}
 & MEP,  PSP 
& 15.00 & 0.27 & 5.30 & 96 & -- \\
 & PP 
& 15.00 & 0.87 & -- & 96 & -- \\
 & HP 
& 17.00 & 0.32 & -- & 67 & -- \\
\hline

\multirow{1}{*}{Other method}
 & DPT 
& 15.4112 & 0.3692 &  4.5033 & -- & 0.0920 \\
\hline

\multirow{5}{*}{Extreme value Mixture models}
& EVMM 
& 15.5215 & 0.4187 & 4.1674 & -- & -- (0.0940) \\
 & M\_EVMM 
& 17.1226 & 0.4392 & 4.5244 & -- & -- (0.0948) \\
 & FEVMM 
& 11.7704 & 0.2414 & 4.9287 & -- & 0.2192 \\
 & EVIMM 
& 15.5211 & 0.4795 & 3.5478 & -- & -- (0.0920) \\
 & EEVIMM 
& 11.7814 & 0.1909 & 5.6067 & -- & 0.1540 \\
\hline
\end{tabular}
\caption{Heavy-tailed scenario; True parameter vector 
$\boldsymbol{\Theta}=(0.3,0.2,2,1,0.5,0.15,11.8831,0.2,5)$; $n=1000$.}
\label{Table-106}
\end{table}

\begin{table}[!t]
\centering
\small
\textbf{Threshold and tail parameter estimates from various methods}
\vspace{2mm}
\begin{tabular}{lllcccc}
\hline
\textbf{Type} & \textbf{Method / Model} & $u$ & $\xi$ & $\sigma$ & $N_u$ &  $\phi_3/\phi_u$ \\
\hline
\multicolumn{2}{l}{\textbf{True value}} 
 & $8.4879$ & $0.00$ & $5.00$ & $150$ & $0.15$ \\
\hline

\multirow{3}{*}{Numerical methods} 
 & $90$th percentile 
& 11.0471 & 0.0591 & 4.7941 & 100& -- \\
 & $k=\sqrt{n}$     
& 16.5622 & 0.0425 & 5.5168 &  31 & -- \\
 & $k=n^{2/3}/\log\log n$ 
& 14.3283 & 0.1261 & 4.5581 &  51 & -- \\
\hline

\multirow{3}{*}{Graphical methods}
 & MEP,  PSP 
& 11.00 & 0.06 & 4.80 & 101 & -- \\
 & PP 
& 11.00  & 0.68 & -- & 101 & -- \\
 & HP
& 13.00 & 0.30 & -- & 68 & -- \\
\hline

\multirow{1}{*}{Other method}
 & DPT 
& 11.6561 &0.1361 &  4.1507 & -- & 0.094 \\
\hline

\multirow{5}{*}{Extreme value Mixture models}
& EVMM 
& 22.3514 & 0.7397 & 1.7237 & -- & -- (0.0119) \\
 & M\_EVMM 
& 11.9681 & 0.1664 & 3.9619 & -- & -- (0.1259) \\
 & FEVMM 
& 8.5230 & 0.0065 & 5.3339 & -- & 0.2164 \\
 & EVIMM 
& 10.6284 & 0.0664 & 4.4296 & -- & -- (0.1090) \\
 & EEVIMM 
& 8.4407 & -0.0012 & 5.4274 & -- & 0.1540 \\
\hline

\end{tabular}
\caption{
Exponential-tailed scenario; True parameter vector 
$\boldsymbol{\Theta}=(0.3,0.2,2,1,0.5,0.15,8.4879,0.00,5)$; $n=1000$.}
\label{Table-107}
\end{table}

\begin{table}[!t]
\centering
\small
\textbf{Threshold and tail parameter estimates from various methods}
\vspace{2mm}
\begin{tabular}{lllcccc}
\hline
\textbf{Type} & \textbf{Method / Model} & $u$ & $\xi$ & $\sigma$ & $N_u$ &  $\phi_3/\phi_u$  \\
\hline
\multicolumn{2}{l}{\textbf{True value}} 
 & $7.7263$ & $-0.20$ & $5.00$ & $150$ & $0.15$ \\
\hline
\multirow{3}{*}{Numerical methods} 
 & $90$th percentile 
& 10.1589 & -0.1530 & 4.3623  &100 & -- \\
 & $k=\sqrt{n}$     
& 14.6264 &-0.1914 & 4.0072 &  31 & --\\
 & $k=n^{2/3}/\log\log n$ 
&  12.9347& -0.1138 & 3.7286 &  51 & -- \\
\hline
\multirow{3}{*}{Graphical methods}
 & MEP,  PSP 
& 10.00 & -0.15 & 4.40 & 104 & -- \\
 & PP 
& 10.00 & 0.64 & -- & 104 & -- \\
 & HP
& 12.00 & 0.25 & -- & 67 & -- \\
\hline

\multirow{1}{*}{Other method}
 & DPT 
& 10.6334 & -0.0938& 3.8336& -- & 0.0950 \\
\hline

\multirow{5}{*}{Extreme value mixture models}
& EVMM 
& 10.9028 & -0.0494 & 3.5199 & -- & -- (0.0910) \\
 & M\_EVMM 
& 10.6426 & -0.0836 & 3.7613 & -- & -- (0.1344) \\
 & FEVMM 
& 7.5634 & -0.2183 & 5.5165 & -- & 0.2207 \\
 & EVIMM 
& 10.903 & -0.1494 & 3.9247 & -- & -- (0.0920) \\
 & EEVIMM 
& 7.6878 & -0.2079 & 5.374 & -- & 0.1540 \\
\hline
\end{tabular}
\caption{Light-tailed scenario; True parameter vector 
$\boldsymbol{\Theta}= (0.3,0.2,2,1,0.5,0.15,7.7263,-0.20,5)$; $n=1000$.}
\label{Table-108}
\end{table}
\section{Analysis of real life data examples}\label{Applications}
\noindent In this section, we illustrate the proposed methodology using two real-world datasets: rainfall data from Hyderabad and Swedish motor insurance data. The threshold and tail parameters are estimated using several classical numerical, graphical, and EVMM-based methods. We further study the associated risk measures for both datasets and present the corresponding results.
\subsubsection*{Illustration 1: Analysis of daily rainfall in Hyderabad, India}
\noindent
This study analyses daily rainfall data for Hyderabad, India, motivated by the need for reliable rainfall models to support urban planning, agriculture, and disaster management. The analysis focuses on the monsoon months (June-September) over the period June 2010 to September 2020, with $n = 1{,}220$ observations.
The population of Hyderabad is growing rapidly, and it is a major financial and business centre, making it vulnerable to extreme rainfall events that can cause economic losses and disrupt daily activities. In addition, dry periods in the monsoon season affect water availability for business purposes and daily life. The rainfall data were obtained from Kaggle (\href{https://www.kaggle.com/datasets/rohanvarma9187/hyderabad-weather1950-2020-temp-and-rainfall}{Hyderabad Daily Weather Data}), and parameter estimates from different methods are reported in Table~\ref{tab:threshold1}.

\begin{table}[!t]
\centering
\small
\textbf{Panel A1: Threshold and tail parameter estimates from various methods.}
\vspace{2mm}
\begin{tabular}{lllcccc}
\hline
\textbf{Type} & \textbf{Method / Model} & $u$ & $\xi$ & $\sigma$ & $N_u$ & $\phi_3/\phi_u$ \\
\hline
\multirow{3}{*}{Numerical methods} 
 & $90$th percentile 
&  12.7366 &  0.0560&  15.6386  & 122 & -- \\

 & $k=\sqrt{n}$     
 & 34.1336 &  0.0442 & 16.5534  &  34& -- \\
 & $k=n^{2/3}/\log\log n$ 
& 23.8152 & -0.0342 & 18.9323  &  58  & -\\
\hline
\multirow{3}{*}{Graphical methods}
 & MEP,  PSP 
& 13.0000 & 0.0560 & 16.0000 & 122 & -- \\
 & PP 
& 13.0000 & 0.8000 & -- & 122 & -- \\
 & HP
& 24.0000 & 0.5000 & -- & 58 & -- \\
\hline

\multirow{1}{*}{Other method}
 & DPT 
& 15.0804 & 0.0269 &16.7063  & -- & 0.0836  \\
\hline
\multirow{5}{*}{Extreme value Mixture models}
& EVMM 
& 8.4379 & 0.1346 & 13.1976 & -- & -- (0.2858) \\
 & M\_EVMM 
& 15.072 & 0.0264 & 16.7327 & -- & -- (0.1777) \\
 & FEVMM 
& 13.6606 & 0.1037 & 14.4022 & -- & 0.2073 \\
 & EVIMM 
& 10.0755 & 0.0428 & 15.5298 & -- & -- (0.1221) \\
 & EEVIMM 
& 7.0608 & 0.1311 & 13.0423 & -- & 0.1590 \\
\hline
\end{tabular}
\label{tab:threshold2}

\vspace{4mm}
\textbf{Panel B1: Additional EEVIMM parameter estimates}
\vspace{2mm}

\begin{tabular}{ccccc}
\hline
$\phi_1$ & $\phi_2$ & $\theta$ & $d$ & $k$\\
\hline
0.5295 & 0.1352 & 0.9835 & 1.2822 & 0.3000\\
\hline
\end{tabular}
\caption{Parameter estimates for the Hyderabad rainfall dataset.
Panel A1 reports threshold and tail parameter estimates $(u,\xi,\sigma,\phi_3)$.
Panel B1 reports additional model-specific parameters.}
\label{tab:threshold1}
\end{table}

\subsubsection*{Illustration 2: Analysis of Swedish motor insurance, Sweden}
\noindent This study analyses the Swedish motor insurance dataset, focusing on the \texttt{Payment} variable, which represents the total claim amount (in Swedish kronor), based on a sample of $n=1237$ observations corresponding to territorial risk classes (Zones~$4$-$7$). Accurate modeling of total claims is essential for premium calculation, capital allocation, reinsurance decisions, and overall risk management, particularly because such decisions depend strongly on tail behaviour.
 The data were obtained from (\href{https://www.kaggle.com/datasets/floser/swedish-motor-insurance}{Swedish Motor Insurance}), and parameter estimates from different methods are reported in Table~\ref{tab:threshold2}.

\begin{table}[!t]
\centering
\small
\textbf{Panel A2: Threshold and tail parameter estimates from various methods.}
\vspace{2mm}
\begin{tabular}{lllcccc}
\hline
\textbf{Type} & \textbf{Method / Model} & $u$ & $\xi$ & $\sigma$ & $N_u$ & $\phi_3$/$\phi_u$ \\
\hline
\multirow{3}{*}{Numerical methods} 
 & $90$th percentile 
&3.5345& 0.7637 & 4.9416&  124 & -- \\
 & $k=\sqrt{n}$     
 & 15.2949 &0.3642 &18.2864&   35 & -- \\
 & $k=n^{2/3}/\log\log n$ 
& 7.7403& 0.4245 &13.6210&   58  & -\\
\hline
\multirow{3}{*}{Graphical methods}
 & MEP,  PSP 
& 3.500 & 0.7600 & 4.9000 & 125 & -- \\
 & PP 
& 3.5000 & 1.5000 & -- & 125 & -- \\
 & HP
& 5.3000 & 0.9700 & -- & 91 & -- \\
\hline

\multirow{1}{*}{Other method}
 & DPT 
& 25.2782 & 0.4553 & 19.8433  & -- & 0.0170  \\
\hline
\multirow{5}{*}{Extreme value Mixture models}
 & EVMM 
& 1.7081 & 0.7238 & 3.8328 & -- & -- (0.1816) \\
 & M\_EVMM 
& 3.1985 & 0.7602 & 4.6681 & -- & -- (0.1455) \\
 & FEVMM 
& 1.232 & 0.8191 & 2.7792 & -- & 0.2538 \\
 & EVIMM 
& 3.0051 & 0.5747 & 6.243 & -- & -- (0.1124) \\
 & EEVIMM 
& 1.7103 & 0.7243 & 3.8173 & -- & 0.1520 \\
\hline
\end{tabular}
\label{tab:threshold2}

\vspace{4mm}
\textbf{Panel B2: Additional EEVIMM parameter estimates}
\vspace{2mm}

\begin{tabular}{ccccc}
\hline
$\phi_1$ & $\phi_2$ & $\theta$ & $d$ & $k$\\
\hline
  0.2611 & 0.0525  & 1.8988  & 0.0178  & 0.3562\\
\hline
\end{tabular}
\caption{Parameter estimates for the Swedish motor insurance dataset.
Panel A2 reports threshold and tail parameter estimates $(u,\xi,\sigma,\phi_3)$.
Panel B2 reports additional model-specific parameters.}
\label{tab:threshold2}
\end{table}
The performance of the proposed EEVIMM against existing considered models ( EVMM, M\_EVMM, FEVMM, EVIMM) is evaluated using both the normalised Akaike Information Criterion (nAIC) (\cite{cohen2021normalized}) and standard GoF tests. Model diagnostics are performed using standard GoF tests, namely the $AD$, $CvM$, and $KS$ statistics. For each considered model, the null hypothesis ($H_0$) in the GoF tests states that the fitted model adequately describes the observed data, i.e., $H_0: F(x) = F_{\hat{\theta}_\mathcal{M}}(x)$, where $\mathcal{M}$ denotes the fitted model and $F_{\hat{\theta}_\mathcal{M}}(x)$ is its CDF with estimated parameters. The corresponding test statistics and associated $p$-values, obtained using a parametric bootstrap with 10{,}000 replications from the fitted model, are reported in Table~\ref{tab:gof_results}. For all GoF tests, the model with a smaller test statistic and a larger \(p\)-value is considered to fit the observed data better, as shown in Table~\ref{tab:gof_results}. Lower AIC values indicate a better model fit, and the EEVIMM achieves the lowest nAIC and the best overall fit among all considered models according to the three GoF tests. Graphical assessments are provided in Figure \ref{fig:return_pp_comparison}; the Figures \ref{fig:pp_hyderabad} and \ref{fig:pp_insurance}
gives the empirical CDF vs. the fitted CDF for the Hyderabad rainfall and insurance data sets, highlighting the visual agreement between the empirical CDF and the fitted EEVIMM. Figure \ref{fig:return_hyderabad} and ~\ref{fig:return_insurance} present the return level plots against the return period, along with the 95\% confidence intervals, used in estimating extreme quantiles relevant for real-world risk assessment.

\begin{table}[!t]
\centering
\setlength{\tabcolsep}{3.5pt}
\renewcommand{\arraystretch}{1.1}
\caption{Goodness-of-fit tests and information criteria for the Hyderabad rainfall and Swedish motor insurance datasets.}
\begin{tabular}{|c|c|c|c|c|}
\hline
\textbf{Model} & $AD$ ($p$-value) & $CvM$ ($p$-value) & $KS$ ($p$-value) & \textbf{nAIC}\\
\hline
\multicolumn{5}{|c|}{\textbf{Hyderabad Rainfall Dataset}} \\
\hline
EVMM  & 7534.2880 (<0.01)  & 119.4132 (<0.01) & 0.5295 (<0.01) & 6.1657 \\
M\_EVMM & 7536.071  (<0.01)  &  115.2901(<0.01)  & 0.5295 (<0.01) & 20.0988 \\
FEVMM & 7534.1770(<0.01)  & 117.5738 (<0.01) & 0.5295 (<0.01) &  6.1393 \\
EVIMM & 274.7258 (0.2991) &  60.55872 (0.2623 ) & 0.0251 (0.2556) &  4.2860 \\
EEVIMM & 274.0181 (0.7668) & 60.40326 (0.6679) & 0.0173 (0.6342) & 3.5920  \\
\hline
\multicolumn{5}{|c|}{\textbf{Swedish Motor Insurance Dataset}} \\
\hline
EVMM & 1862.0195 (<0.01) & 34.9961 (<0.01) & 0.2611 (<0.01) & 2.2962  \\
M\_EVMM & 1858.3208 (<0.01) & 30.4200 (<0.01) & 0.2611 (<0.01) & 16.2301 \\
FEVMM & 1859.9596 (<0.01) & 31.9663 (<0.01) & 0.2611 (<0.01) & 2.2487  \\
EVIMM & 56.4606 (<0.01) & 8.5508 (<0.01) & 0.0733 (<0.01) & 2.9052  \\
EEVIMM & 51.5128 (0.9226) & 7.3760 (0.8696) & 0.0228 (0.4806) & 2.0388  \\
\hline
\end{tabular}
\label{tab:gof_results}
\end{table}

\begin{figure}[!t]
    \centering
    \begin{subfigure}{0.48\textwidth}
        \centering
        \includegraphics[width=7.5cm,height=7cm]{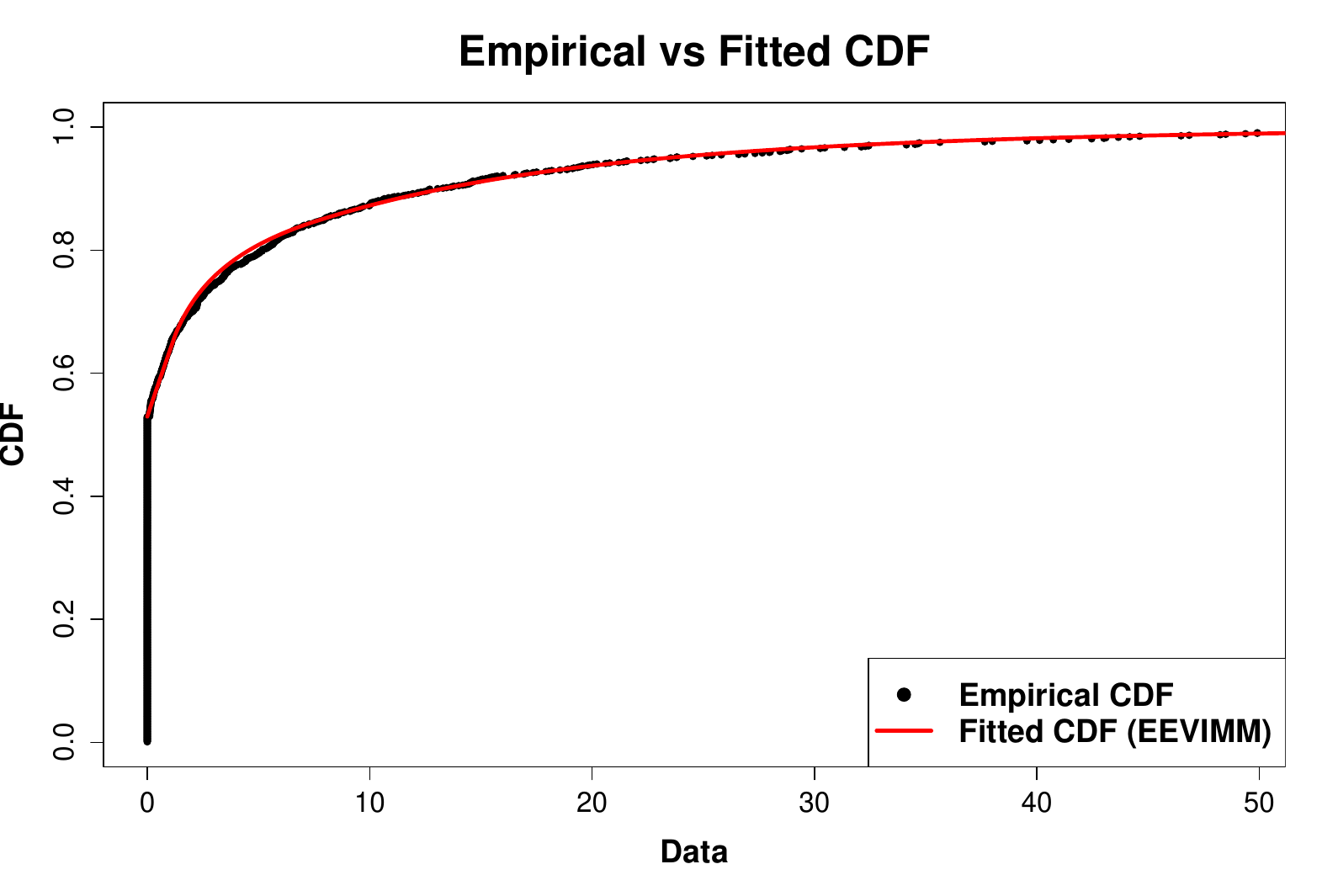}
       \caption{Empirical vs fitted CDF for the Hyderabad rainfall data.}
        \label{fig:pp_hyderabad}
    \end{subfigure}
    \hfill
    \begin{subfigure}{0.48\textwidth}
        \centering
        \includegraphics[width=7.5cm,height=7cm]{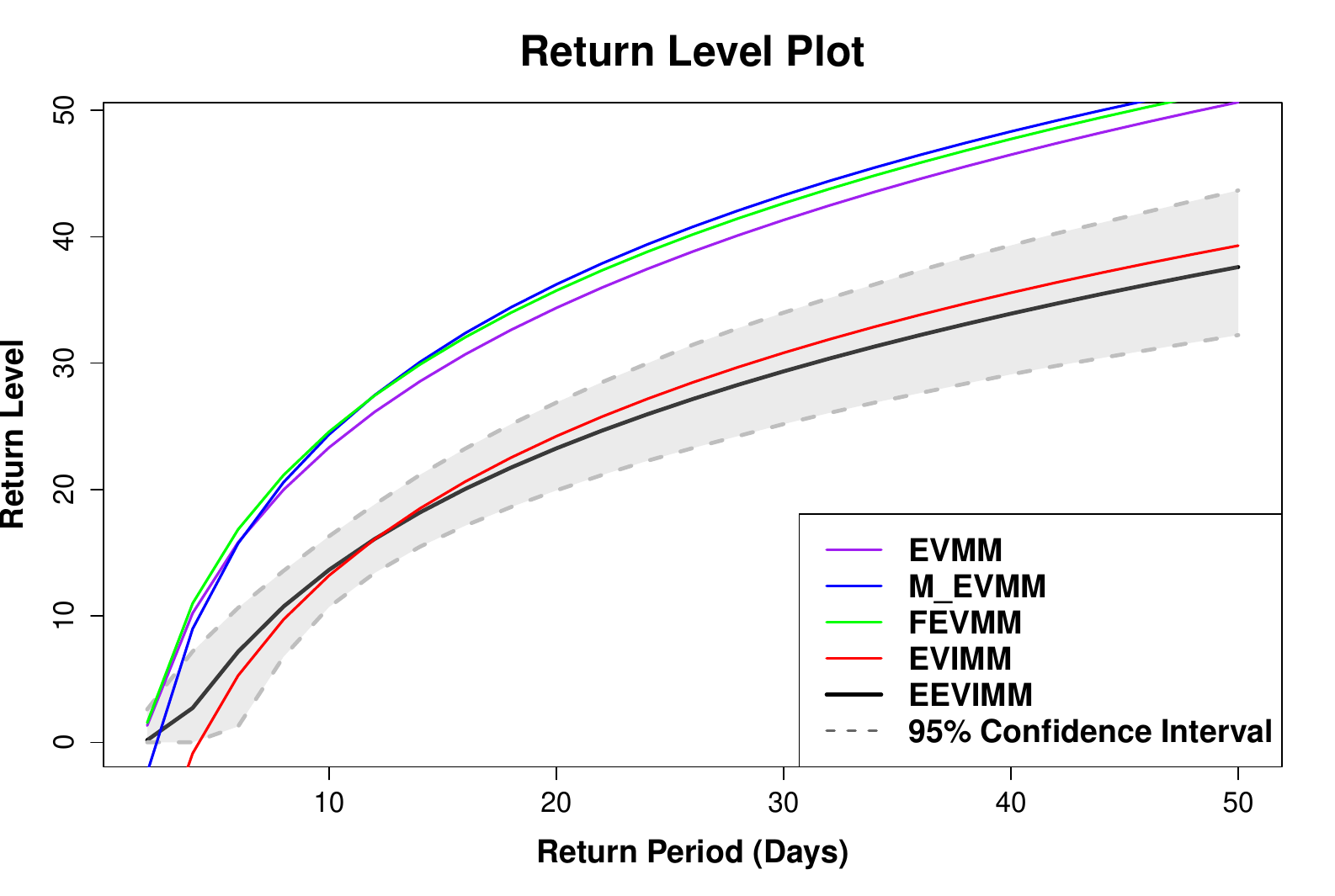}
      \caption{Return level estimates for the Hyderabad rainfall data.}
        \label{fig:return_hyderabad}
    \end{subfigure}

    \vspace{0.4cm}
    
    \begin{subfigure}{0.48\textwidth}
        \centering
        \includegraphics[width=7.5cm,height=7cm]{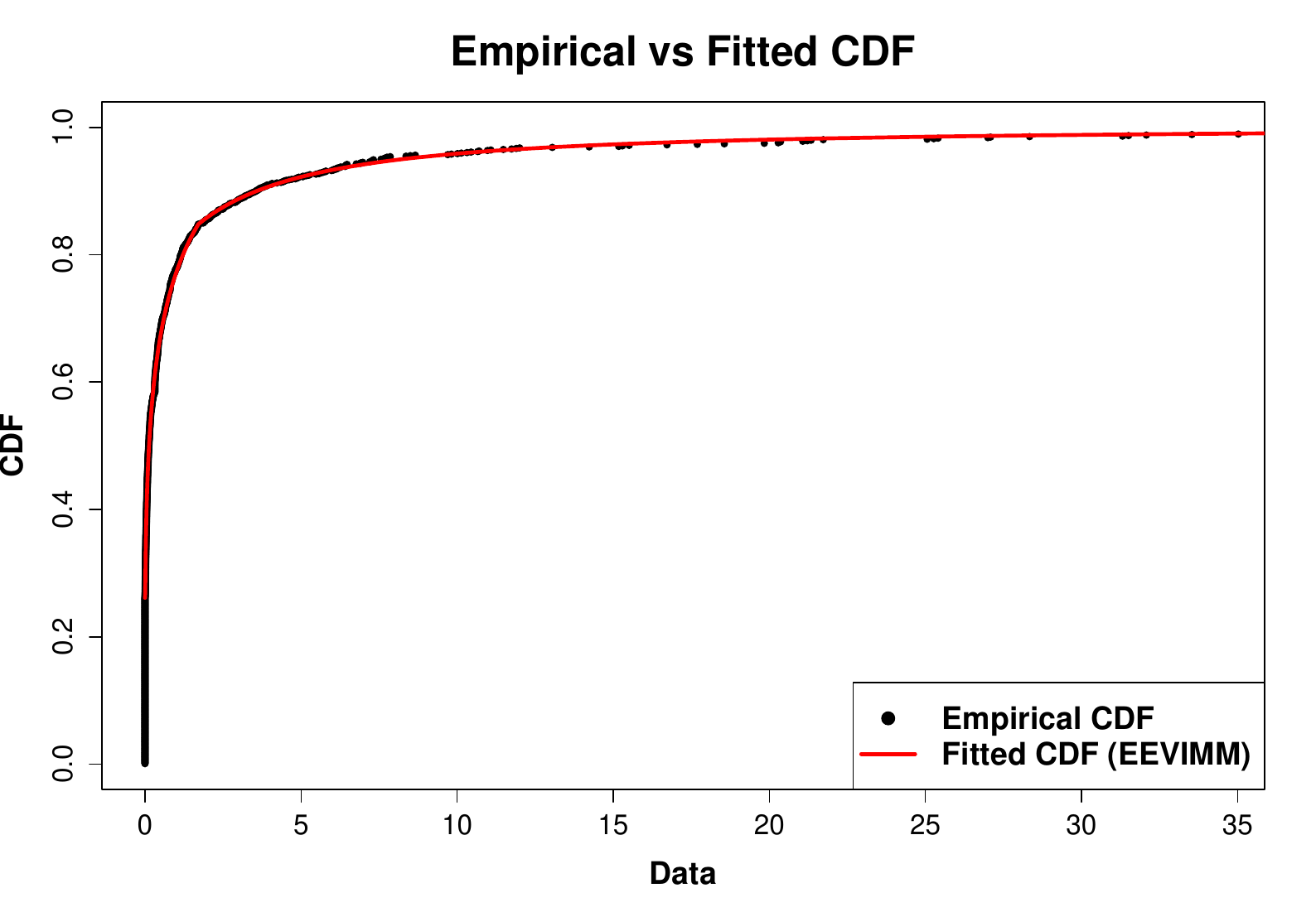}
       \caption{Empirical vs fitted CDF for the Swedish motor insurance data.}
        \label{fig:pp_insurance}
    \end{subfigure}
    \hfill
    \begin{subfigure}{0.48\textwidth}
        \centering
        \includegraphics[width=7.5cm,height=7cm]{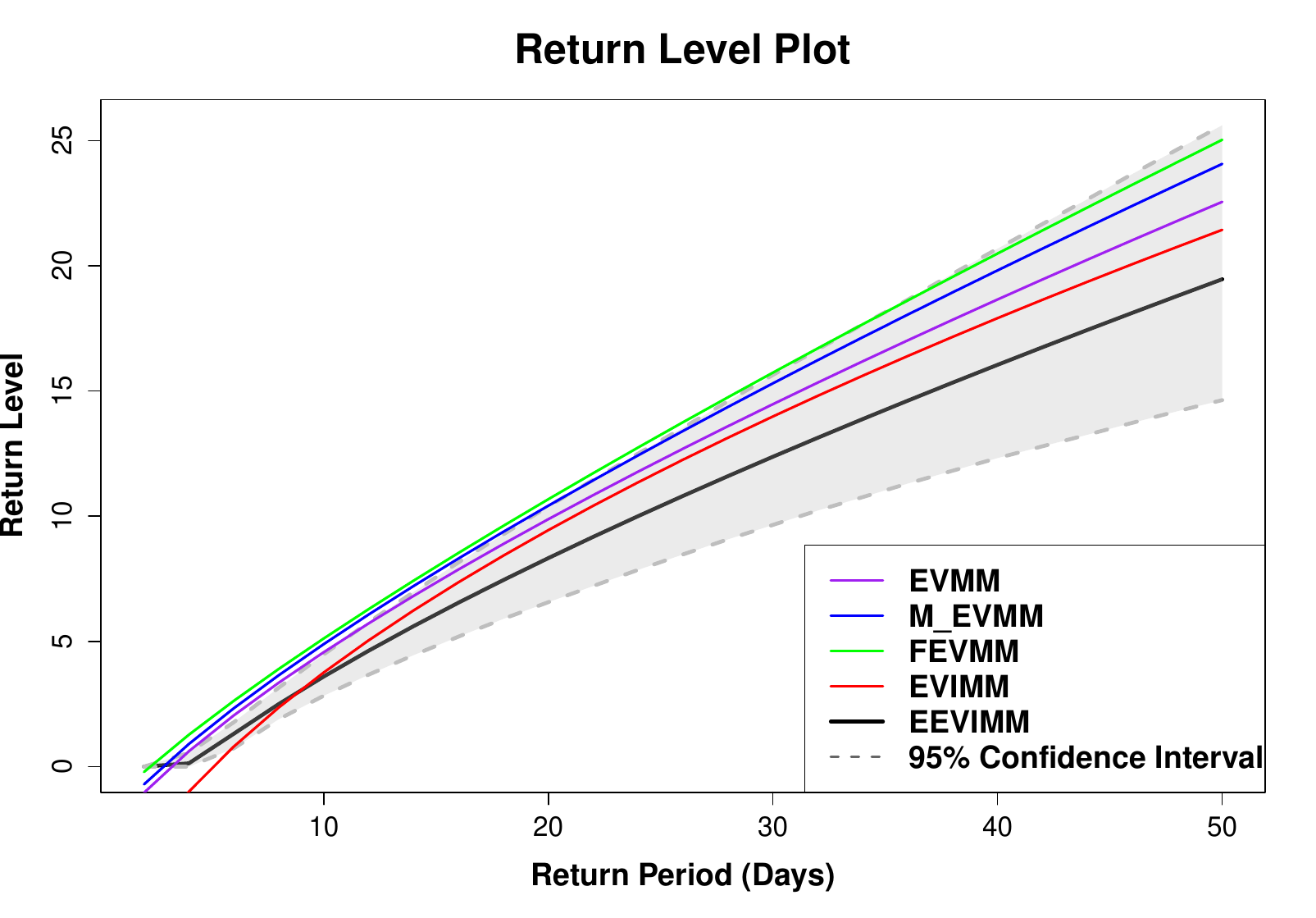}
         \caption{Return level estimates for the Swedish motor insurance data.}
        \label{fig:return_insurance}
    \end{subfigure}
 \caption{Empirical vs fitted CDF plots and return level estimates for the Hyderabad rainfall and Swedish motor insurance datasets.}
    \label{fig:return_pp_comparison}
\end{figure}

\section{Conclusions and future directions}\label{conclusion}
\noindent In this paper, we focused on better modeling of extreme events in the presence of instantaneous and early failures by proposing a novel framework, EEVIMM. Estimating an appropriate threshold is a key and persistent problem in extreme value analysis, as it affects the estimation of tail parameters. The presence of inliers in the data below the threshold further affects these parameters (see \cite{nila2025modeling} and Figure~\ref{fig:meanexcess}).
Existing classical threshold estimation methods can produce distorted threshold and tail estimates in the presence of inliers; addressing this issue was the main objective of this study. 
In general, it is assumed that instantaneous and early failures do not significantly affect the estimation of extreme-value parameters, such as the threshold or tail fraction. However, simulation studies and real-data applications show that ignoring these components can significantly affect extreme-parameter estimation. This is further supported by the theoretical finding that the mixture proportions $\phi_1$, $\phi_2$, and $\phi_3$ are inherently dependent under the unit-sum constraint, resulting in asymptotic dependence among their corresponding MLEs. Accounting for this dependence improves inference for the extreme component and its associated risk measures.
The performance of the proposed methodology was evaluated through extensive Monte Carlo simulation studies (Tables~\ref{Table-100}--\ref{Table-103}); EEVIMM generally produced estimates closer to the true parameter values and provided more accurate tail estimation than classical numerical methods, graphical methods such as MEP, PSP, PP, and HP, and EVMMs, including the recently proposed EVIMM (\cite{nila2025modeling}) and the DPT method (\cite{sakthivel2025dual}). In real data analysis, GoF results, nAIC, test statistic values, and $p$-values (see Table~\ref{tab:gof_results}), together with graphical diagnostics (Figures~\ref{fig:pp_hyderabad} and \ref{fig:pp_insurance}), indicate that it provides a better overall fit. 
The motivation for this work comes from real-world data, where inliers can influence extreme value analysis. For example, in rainfall data, observations may include no rainfall and low, moderate, and heavy intensities. Both very low and very high rainfall can lead to serious events such as droughts or floods. Several directions for future work can be considered. One important direction is to further investigate how the choice of the distribution below the threshold affects both the threshold and the extreme-value parameters. In addition, there is no general framework for defining EVMMs \citet{behrens2004bayesian,hu2018evmix,hu2013extreme}, and studies that explicitly incorporate inliers below the threshold remain limited \citet{nila2025modeling}. The proposed framework helps address this gap.
One challenge is that optimisation has not been well explored in the literature for such mixture models. This suggests that alternative estimation procedures, such as the Bayesian approach, can be considered to improve stability and estimation accuracy.
 One possible future direction for work is to consider Bernoulli-type events at multiple points, along with observations of extreme and other intensities. For greater flexibility in the bulk part, mixture distributions can be used together with instantaneous failure and early-range components. A Dirichlet process approach may also be used to allow more flexibility in modeling different data structures.  Finally, applying the proposed model to additional real-world data in reliability, environmental studies, and finance, and comparing different estimation methods, can further enhance its practical usefulness.
Overall, the proposed model provides a flexible and robust framework for modeling extreme events in the presence of instantaneous and early failures, leading to improved threshold estimation and better overall performance.
 \bibliographystyle{apalike}
 \bibliography{01ref}
 \section*{Appendix A. Proofs of Propositions}
\subsection*{Proposition 1}
\begin{proof}
From \eqref{eq:EEVIMM_general_pdf}, the limits at \(d\) are
\begin{align*}
f(d^-)
&=\phi_2\frac{f_1(d|\Phi_1)}{F_1(d|\Phi_1)},
&f(d^+)
&=(1-\phi_1-\phi_2-\phi_3)
\frac{f_2(d|\Phi_2)}
{F_2(u|\Phi_2)-F_2(d|\Phi_2)}.
\end{align*}
Hence, continuity at \(d\), i.e., \(f(d^-)=f(d^+)\), gives
\[
\phi_2 \frac{f_1(d|\Phi_1)}{F_1(d|\Phi_1)}
=
(1-\phi_1-\phi_2-\phi_3)
\frac{f_2(d|\Phi_2)}
{F_2(u|\Phi_2)-F_2(d|\Phi_2)}.
\]
For the DF in \eqref{evimmpdf}, the limits at \(d\) are
\begin{align*}
f(d^-) &=
\phi_2 \frac{\theta}{d},
&
f(d^+)
&=
(1-\phi_1-\phi_2-\phi_3)
\frac{
kd^{k-1}e^{-d^k}
}{
e^{-d^k}-e^{-u^k}
}.
\end{align*}
Therefore,
\begin{align*}
\phi_2 \frac{\theta}{d}
&=
(1-\phi_1-\phi_2-\phi_3)
\frac{
kd^{k-1}e^{-d^k}
}{
e^{-d^k}-e^{-u^k}
}
\\
&\iff
\theta
=
\frac{d}{\phi_2}
(1-\phi_1-\phi_2-\phi_3)
\frac{
kd^{k-1}e^{-d^k}
}{
e^{-d^k}-e^{-u^k}
}.
\end{align*}

Similarly, the limits at \(u\) from \eqref{eq:EEVIMM_general_pdf} are
\[f(u^-)=(1-\phi_1-\phi_2-\phi_3)
\frac{f_2(u|\Phi_2)}
{F_2(u|\Phi_2)-F_2(d|\Phi_2)},
\qquad f(u^+)
=\phi_3 g(u|u,\xi,\sigma)
=\frac{\phi_3}{\sigma}.\]
Thus,\[(1-\phi_1-\phi_2-\phi_3)
\frac{f_2(u|\Phi_2)}
{F_2(u|\Phi_2)-F_2(d|\Phi_2)}
=
\frac{\phi_3}{\sigma}
\iff
\sigma =\frac{
\phi_3\{F_2(u|\Phi_2)-F_2(d|\Phi_2)\}
}{
(1-\phi_1-\phi_2-\phi_3)f_2(u|\Phi_2)
}.
\]
For the specific DF in \eqref{evimmpdf}, the limits at \(u\) reduce to
\begin{align*}
f(u^-)
&=
(1-\phi_1-\phi_2-\phi_3)
\frac{
ku^{k-1}e^{-u^k}
}{
e^{-d^k}-e^{-u^k}
},
&
f(u^+)
&=
\phi_3 g(u\mid u,\xi,\sigma)
=
\frac{\phi_3}{\sigma}.
\end{align*}
Hence,
\begin{align*}
(1-\phi_1-\phi_2-\phi_3)
\frac{
ku^{k-1}e^{-u^k}
}{
e^{-d^k}-e^{-u^k}
}
&=
\frac{\phi_3}{\sigma}
\\
&\iff
\sigma
=
\frac{
\phi_3
\left(
e^{-d^k}-e^{-u^k}
\right)
}{
(1-\phi_1-\phi_2-\phi_3)
ku^{k-1}e^{-u^k}
}.
\end{align*}
Hence, the DF in \eqref{evimmpdf} is continuous at both \(d\) and \(u\) if and only if the stated conditions hold.
This completes the proof.
\end{proof}
\subsection*{Proposition 2}
\begin{proof}
First, assume that Proposition~\ref{rem:Continuity} holds. From \eqref{eq:EEVIMM_general_pdf}, evaluating derivatives at \(d\),
\begin{align*}
f'(d^-)
&=
\phi_2\frac{f_1'(d|\Phi_1)}{F_1(d|\Phi_1)},
&
f'(d^+)
&=
(1-\phi_1-\phi_2-\phi_3)
\frac{f_2'(d|\Phi_2)}
{F_2(u|\Phi_2)-F_2(d|\Phi_2)}.
\end{align*}
Thus, differentiability at \(d\), i.e., \(f'(d^-)=f'(d^+)\), gives
\[
\phi_2\frac{f_1'(d|\Phi_1)}{F_1(d|\Phi_1)}
=
(1-\phi_1-\phi_2-\phi_3)
\frac{f_2'(d|\Phi_2)}
{F_2(u|\Phi_2)-F_2(d|\Phi_2)}.
\]
For the DF in \eqref{evimmpdf}, at \(x=d\), we have
\[
f_1(x)=\frac{\theta}{d^\theta}x^{\theta-1}
\implies
f_1'(d)=\frac{\theta(\theta-1)}{d^2}.
\]

\[
f_2(x)=kx^{k-1}e^{-x^k},
\implies
f_2'(d)
=
f_2(d)\left(\frac{k-1}{d}-kd^{k-1}\right).
\]
Here \(F_1(d)=1\), and we obtain
\[
\phi_2\frac{\theta(\theta-1)}{d^2}
=
(1-\phi_1-\phi_2-\phi_3)
\frac{
f_2(d)\left(\frac{k-1}{d}-kd^{k-1}\right)
}{
e^{-d^k}-e^{-u^k}
}.
\]

To establish differentiability at \(u\), the left- and right-hand derivatives are
\begin{align*}
f'(u^-)
&=
(1-\phi_1-\phi_2-\phi_3)
\frac{f_2'(u|\Phi_2)}
{F_2(u|\Phi_2)-F_2(d|\Phi_2)},
&
f'(u^+)
&=
\phi_3 g'(u|u,\sigma,\xi).
\end{align*}
Here \(f_1'(\cdot|\Phi_1)\), \(f_2'(\cdot|\Phi_2)\), and \(g'(\cdot|u,\sigma,\xi)\) denote derivatives with respect to \(x\).
From \eqref{evimmpdf},
\[
g'(u|u,\sigma,\xi)
=
-\frac{1+\xi}{\sigma^2}
\implies
f'(u^+)
=
-\frac{\phi_3(1+\xi)}{\sigma^2}.
\]
Hence, differentiability at \(u\), i.e., \(f'(u^-)=f'(u^+)\), gives
\[
(1-\phi_1-\phi_2-\phi_3)
\frac{f_2'(u|\Phi_2)}
{F_2(u|\Phi_2)-F_2(d|\Phi_2)}
=
-\frac{\phi_3(1+\xi)}{\sigma^2},
\]
and hence
\[
\sigma^2
=
-\frac{
\phi_3(1+\xi)\{F_2(u|\Phi_2)-F_2(d|\Phi_2)\}
}{
(1-\phi_1-\phi_2-\phi_3)f_2'(u|\Phi_2)
}.
\]
For the specific DF in \eqref{evimmpdf}, at \(x=u\),
\[
f_2'(u)
=
f_2(u)
\left(
\frac{k-1}{u}-ku^{k-1}
\right)
\implies
f'(u^-)
=
(1-\phi_1-\phi_2-\phi_3)
\frac{
f_2(u)
\left(\frac{k-1}{u}-ku^{k-1}\right)
}{
e^{-d^k}-e^{-u^k}
}.
\]
Hence,
\[
(1-\phi_1-\phi_2-\phi_3)
\frac{
f_2(u)
\left(\frac{k-1}{u}-ku^{k-1}\right)
}{
e^{-d^k}-e^{-u^k}
}
=
-\frac{\phi_3(1+\xi)}{\sigma^2}.
\]
Solving for \(\sigma^2\), we obtain
\[
\sigma^2
=
-\frac{
\phi_3(1+\xi)
\left(e^{-d^k}-e^{-u^k}\right)
}{
(1-\phi_1-\phi_2-\phi_3)
f_2(u)
\left(\frac{k-1}{u}-ku^{k-1}\right)
}.
\]
Hence, the DF in \eqref{evimmpdf} is differentiable at both \(d\) and \(u\) if and only if the stated conditions hold.
This completes the proof.
\end{proof}
\section*{Appendix B: Additional simulation tables}
\noindent The additional Monte Carlo simulation results are reported in Tables~\ref{Table-EEVIMM}--\ref{Table-EEVIMM2}.
\label{app:bias_mse_comparison_apendex}
\begin{landscape}
\begin{table*}[htbp]
\centering
\scriptsize
\setlength{\tabcolsep}{3pt}
\renewcommand{\arraystretch}{0.9}
  \caption{Monte Carlo simulation results ($N=2000$) for the heavy-tailed scenario with sample sizes $n=300,400,500,750,1000,$ and $1200$, reporting Sample mean, BSE, 95\% BCI, RMSE, Bias, and CP.}
\begin{tabular}{|p{2.5cm}|p{1.8cm}|p{1.8cm}|p{2.3cm}|p{1.8cm}|p{2.6cm}|p{2.3cm}|p{2.3cm}|p{1.5cm}|}
\hline
\textbf{n} & \textbf{Parameters} & \textbf{True Value} &
\textbf{Sample Mean} & \textbf{BSE} & \textbf{BCI} &
\textbf{RMSE} & \textbf{Bias} & \textbf{CP} \\ \hline

\multirow{9}{*}{ 300 }
 & $\phi_1$ &  0.3  &  0.2996  &  0.0275  & ( 0.2767 ,  0.3867 )
 &  0.0260  &  -0.0004  &  90.0  \\ \cline{2-9}
 & $\phi_2$ &  0.2  &  0.1823  &  0.0531  & ( 0.0633 ,  0.2467 )
 &  0.0571  &  -0.0177  &  100.0  \\ \cline{2-9}
 & $\theta$ &  2  &  2.1186  &  0.3941  & ( 1.3285 ,  2.8666 )
 &  0.4695  &  0.1186  &  100.0  \\ \cline{2-9}
 & $d$ &  1  &  0.9436  &  0.1854  & ( 0.4534 ,  1.1649 )
 &  0.1863  &  -0.0564  &  100.0  \\ \cline{2-9}
 & $k$ &  0.5  &  0.4842  &  0.1455  & ( 0.3000 ,  0.8610 )
 &  0.1454  &  -0.0158  &  100.0  \\ \cline{2-9}
 & $\phi_3$ &  0.1  &  0.1233  &  0.0178  & ( 0.1000 ,  0.1600 )
 &  0.0288  &  0.0233  &  96.7  \\ \cline{2-9}
 & $u$ &  7.7263  &  6.5817  &  1.8457  & ( 3.6812 ,  9.3830 )
 &  1.6983  &  -1.1446  &  90.0  \\ \cline{2-9}
 & $\xi$ &  0.2  &  0.1358  &  0.2645  & ( -0.6938 ,  0.3678 )
 &  0.2845  &  -0.0642  &  91.7  \\ \cline{2-9}
 & $\sigma$ &  5  &  5.3609  &  2.1863  & ( 2.5539 ,  10.6416 )
 &  2.0072  &  0.3609  &  96.7  \\ \hline

\multirow{9}{*}{ 400 }
 & $\phi_1$ &  0.3  &  0.3002  &  0.0224  & ( 0.2525 ,  0.3401 )
 &  0.0229  &  0.0002  &  93.3  \\ \cline{2-9}
 & $\phi_2$ &  0.2  &  0.1817  &  0.0548  & ( 0.0675 ,  0.2575 )
 &  0.0536  &  -0.0183  &  100.0  \\ \cline{2-9}
 & $\theta$ &  2  &  2.0840  &  0.3305  & ( 1.4775 ,  2.7205 )
 &  0.3650  &  0.0840  &  98.3  \\ \cline{2-9}
 & $d$ &  1  &  0.9452  &  0.1952  & ( 0.5787 ,  1.3071 )
 &  0.1737  &  -0.0548  &  98.3  \\ \cline{2-9}
 & $k$ &  0.5  &  0.4862  &  0.1081  & ( 0.3000 ,  0.7266 )
 &  0.1260  &  -0.0138  &  100.0  \\ \cline{2-9}
 & $\phi_3$ &  0.1  &  0.1104  &  0.0233  & ( 0.0750 ,  0.1600 )
 &  0.0258  &  0.0104  &  100.0  \\ \cline{2-9}
 & $u$ &  7.7263  &  7.3324  &  1.3351  & ( 5.1886 ,  10.9339 )
 &  1.5641  &  -0.3939  &  100.0  \\ \cline{2-9}
 & $\xi$ &  0.2  &  0.1619  &  0.2535  & ( -0.2304 ,  0.7658 )
 &  0.2838  &  -0.0381  &  90.0  \\ \cline{2-9}
 & $\sigma$ &  5  &  5.2565  &  1.6052  & ( 2.3617 ,  8.7825 )
 &  1.9713  &  0.2565  &  93.3  \\ \hline

\multirow{9}{*}{ 500 }
 & $\phi_1$ &  0.3  &  0.3000  &  0.0210  & ( 0.2760 ,  0.3620 )
 &  0.0210  &  -0.0000  &  98.3  \\ \cline{2-9}
 & $\phi_2$ &  0.2  &  0.1829  &  0.0502  & ( 0.0660 ,  0.2460 )
 &  0.0515  &  -0.0171  &  100.0  \\ \cline{2-9}
 & $\theta$ &  2  &  2.0518  &  0.2946  & ( 1.6189 ,  2.7493 )
 &  0.3047  &  0.0518  &  96.7  \\ \cline{2-9}
 & $d$ &  1  &  0.9491  &  0.1537  & ( 0.5709 ,  1.1471 )
 &  0.1640  &  -0.0509  &  96.7  \\ \cline{2-9}
 & $k$ &  0.5  &  0.4849  &  0.1210  & ( 0.3000 ,  0.7896 )
 &  0.1125  &  -0.0151  &  100.0  \\ \cline{2-9}
 & $\phi_3$ &  0.1  &  0.1020  &  0.0237  & ( 0.0720 ,  0.1580 )
 &  0.0277  &  0.0020  &  100.0  \\ \cline{2-9}
 & $u$ &  7.7263  &  7.8835  &  1.3402  & ( 4.0550 ,  9.0470 )
 &  1.7836  &  0.1572  &  100.0  \\ \cline{2-9}
 & $\xi$ &  0.2  &  0.1716  &  0.2795  & ( -0.0484 ,  1.0402 )
 &  0.2740  &  -0.0284  &  91.7  \\ \cline{2-9}
 & $\sigma$ &  5  &  5.2894  &  1.2963  & ( 2.0004 ,  7.3075 )
 &  2.1019  &  0.2894  &  93.3  \\ \hline

\multirow{9}{*}{ 750 }
 & $\phi_1$ &  0.3  &  0.3000  &  0.0168  & ( 0.2787 ,  0.3427 )
 &  0.0168  &  0.0000  &  90.0  \\ \cline{2-9}
 & $\phi_2$ &  0.2  &  0.1795  &  0.0469  & ( 0.0627 ,  0.2320 )
 &  0.0515  &  -0.0205  &  91.7  \\ \cline{2-9}
 & $\theta$ &  2  &  2.0452  &  0.2757  & ( 1.8320 ,  2.8682 )
 &  0.2508  &  0.0452  &  98.3  \\ \cline{2-9}
 & $d$ &  1  &  0.9410  &  0.1338  & ( 0.5793 ,  1.0670 )
 &  0.1593  &  -0.0590  &  90.0  \\ \cline{2-9}
 & $k$ &  0.5  &  0.4965  &  0.0682  & ( 0.3000 ,  0.5255 )
 &  0.0955  &  -0.0035  &  98.3  \\ \cline{2-9}
 & $\phi_3$ &  0.1  &  0.1013  &  0.0152  & ( 0.0707 ,  0.1320 )
 &  0.0292  &  0.0013  &  98.3  \\ \cline{2-9}
 & $u$ &  7.7263  &  7.9751  &  0.9576  & ( 6.0179 ,  10.3996 )
 &  1.8426  &  0.2488  &  100.0  \\ \cline{2-9}
 & $\xi$ &  0.2  &  0.1886  &  0.1656  & ( -0.6971 ,  -0.0460 )
 &  0.2195  &  -0.0114  &  93.3  \\ \cline{2-9}
 & $\sigma$ &  5  &  5.0928  &  2.3462  & ( 6.7524 ,  15.7970 )
 &  1.6796  &  0.0928  &  95.0  \\ \hline

\multirow{9}{*}{ 1000 }
 & $\phi_1$ &  0.3  &  0.2997  &  0.0136  & ( 0.2250 ,  0.2780 )
 &  0.0141  &  -0.0003  &  93.3  \\ \cline{2-9}
 & $\phi_2$ &  0.2  &  0.1779  &  0.0478  & ( 0.0730 ,  0.2510 )
 &  0.0513  &  -0.0221  &  90.0  \\ \cline{2-9}
 & $\theta$ &  2  &  2.0325  &  0.2049  & ( 1.7642 ,  2.6147 )
 &  0.2066  &  0.0325  &  95.0  \\ \cline{2-9}
 & $d$ &  1  &  0.9346  &  0.1436  & ( 0.6214 ,  1.1678 )
 &  0.1572  &  -0.0654  &  91.7  \\ \cline{2-9}
 & $k$ &  0.5  &  0.4980  &  0.0769  & ( 0.4038 ,  0.7166 )
 &  0.0843  &  -0.0020  &  100.0  \\ \cline{2-9}
 & $\phi_3$ &  0.1  &  0.0999  &  0.0230  & ( 0.0730 ,  0.1600 )
 &  0.0283  &  -0.0001  &  100.0  \\ \cline{2-9}
 & $u$ &  7.7263  &  8.0334  &  1.2337  & ( 4.9241 ,  9.8891 )
 &  1.8046  &  0.3071  &  100.0  \\ \cline{2-9}
 & $\xi$ &  0.2  &  0.1958  &  0.1163  & ( -0.3149 ,  0.1461 )
 &  0.1812  &  -0.0042  &  96.7  \\ \cline{2-9}
 & $\sigma$ &  5  &  5.0716  &  0.9769  & ( 3.7557 ,  7.6231 )
 &  1.3978  &  0.0716  &  98.3  \\ \hline

\multirow{9}{*}{ 1200 }
 & $\phi_1$ &  0.3  &  0.2995  &  0.0130  & ( 0.2692 ,  0.3217 )
 &  0.0129  &  -0.0005  &  90.0  \\ \cline{2-9}
 & $\phi_2$ &  0.2  &  0.1763  &  0.0464  & ( 0.0658 ,  0.2333 )
 &  0.0517  &  -0.0237  &  91.7  \\ \cline{2-9}
 & $\theta$ &  2  &  2.0288  &  0.2119  & ( 1.9136 ,  2.7530 )
 &  0.1903  &  0.0288  &  98.3  \\ \cline{2-9}
 & $d$ &  1  &  0.9298  &  0.1269  & ( 0.5998 ,  1.0549 )
 &  0.1603  &  -0.0702  &  90.0  \\ \cline{2-9}
 & $k$ &  0.5  &  0.5030  &  0.0645  & ( 0.3846 ,  0.6552 )
 &  0.0772  &  0.0030  &  98.3  \\ \cline{2-9}
 & $\phi_3$ &  0.1  &  0.1015  &  0.0184  & ( 0.0750 ,  0.1558 )
 &  0.0269  &  0.0015  &  98.3  \\ \cline{2-9}
 & $u$ &  7.7263  &  7.9388  &  1.0188  & ( 5.2964 ,  9.9473 )
 &  1.7146  &  0.2125  &  100.0  \\ \cline{2-9}
 & $\xi$ &  0.2  &  0.2008  &  0.1330  & ( 0.0269 ,  0.5601 )
 &  0.1585  &  0.0008  &  95.0  \\ \cline{2-9}
 & $\sigma$ &  5  &  5.0084  &  0.7881  & ( 3.4211 ,  6.5051 )
 &  1.2189  &  0.0084  &  96.7  \\ \hline
\end{tabular}
\label{Table-EEVIMM}
\end{table*}
\end{landscape}

\begin{landscape}
\begin{table*}[htbp]
\centering
\scriptsize
\setlength{\tabcolsep}{3pt}
\renewcommand{\arraystretch}{0.9}
  \caption{Monte Carlo simulation results ($N=2000$) for the light-tailed scenario with sample sizes $n=300,400,500,750,1000,$ and $1200$, reporting Sample mean, BSE, 95\% BCI, RMSE, Bias, and CP.}
\begin{tabular}{|p{2.5cm}|p{1.8cm}|p{1.8cm}|p{2.3cm}|p{1.8cm}|p{2.6cm}|p{2.3cm}|p{2.3cm}|p{1.5cm}|}
\hline
\textbf{n} & \textbf{Parameters} & \textbf{True Value} &
\textbf{Sample Mean} & \textbf{BSE} & \textbf{BCI} &
\textbf{RMSE} & \textbf{Bias} & \textbf{CP} \\ \hline

\multirow{9}{*}{ 300 }
 & $\phi_1$ &  0.3  &  0.3000  &  0.0257  & ( 0.2467 ,  0.3467 )
 &  0.0267  &  0.0000  &  98.3  \\ \cline{2-9}
 & $\phi_2$ &  0.2  &  0.1838  &  0.0563  & ( 0.0667 ,  0.2567 )
 &  0.0538  &  -0.0162  &  100.0  \\ \cline{2-9}
 & $\theta$ &  2  &  2.1155  &  0.3249  & ( 1.2036 ,  2.4472 )
 &  0.4391  &  0.1155  &  98.3  \\ \cline{2-9}
 & $d$ &  1  &  0.9478  &  0.2225  & ( 0.4598 ,  1.2900 )
 &  0.1790  &  -0.0522  &  98.3  \\ \cline{2-9}
 & $k$ &  0.5  &  0.4675  &  0.1583  & ( 0.3000 ,  0.8265 )
 &  0.1446  &  -0.0325  &  100.0  \\ \cline{2-9}
 & $\phi_3$ &  0.1  &  0.1226  &  0.0193  & ( 0.1000 ,  0.1600 )
 &  0.0282  &  0.0226  &  96.7  \\ \cline{2-9}
 & $u$ &  7.7263  &  6.5217  &  1.0273  & ( 4.2881 ,  8.3593 )
 &  1.7300  &  -1.2046  &  91.7  \\ \cline{2-9}
 & $\xi$ &  -0.2  &  -0.3033  &  0.2305  & ( -0.7458 ,  0.1733 )
 &  0.2615  &  -0.1033  &  88.3  \\ \cline{2-9}
 & $\sigma$ &  5  &  5.9964  &  1.5830  & ( 2.4146 ,  8.3860 )
 &  2.2366  &  0.9964  &  95.0  \\ \hline

\multirow{9}{*}{ 400 }
 & $\phi_1$ &  0.3  &  0.3000  &  0.0230  & ( 0.2774 ,  0.3650 )
 &  0.0229  &  0.0000  &  95.0  \\ \cline{2-9}
 & $\phi_2$ &  0.2  &  0.1834  &  0.0476  & ( 0.0625 ,  0.2375 )
 &  0.0534  &  -0.0166  &  96.7  \\ \cline{2-9}
 & $\theta$ &  2  &  2.0887  &  0.3668  & ( 1.5646 ,  3.0266 )
 &  0.3830  &  0.0887  &  93.3  \\ \cline{2-9}
 & $d$ &  1  &  0.9479  &  0.1701  & ( 0.5787 ,  1.2440 )
 &  0.1737  &  -0.0521  &  93.3  \\ \cline{2-9}
 & $k$ &  0.5  &  0.4751  &  0.1115  & ( 0.3000 ,  0.6515 )
 &  0.1303  &  -0.0249  &  100.0  \\ \cline{2-9}
 & $\phi_3$ &  0.1  &  0.1118  &  0.0260  & ( 0.0750 ,  0.1600 )
 &  0.0262  &  0.0118  &  100.0  \\ \cline{2-9}
 & $u$ &  7.7263  &  7.1634  &  1.3714  & ( 4.7938 ,  10.0252 )
 &  1.5860  &  -0.5629  &  100.0  \\ \cline{2-9}
 & $\xi$ &  -0.2  &  -0.2669  &  0.2298  & ( -0.7674 ,  0.1463 )
 &  0.2393  &  -0.0669  &  93.3  \\ \cline{2-9}
 & $\sigma$ &  5  &  5.5506  &  1.6203  & ( 2.2735 ,  8.1233 )
 &  1.9504  &  0.5506  &  95.0  \\ \hline

\multirow{9}{*}{ 500 }
 & $\phi_1$ &  0.3  &  0.3007  &  0.0199  & ( 0.2660 ,  0.3460 )
 &  0.0205  &  0.0007  &  93.3  \\ \cline{2-9}
 & $\phi_2$ &  0.2  &  0.1800  &  0.0424  & ( 0.0540 ,  0.2100 )
 &  0.0532  &  -0.0200  &  96.7  \\ \cline{2-9}
 & $\theta$ &  2  &  2.0724  &  0.4493  & ( 1.8799 ,  3.6006 )
 &  0.3258  &  0.0724  &  95.0  \\ \cline{2-9}
 & $d$ &  1  &  0.9414  &  0.1155  & ( 0.5221 ,  0.9781 )
 &  0.1665  &  -0.0586  &  91.7  \\ \cline{2-9}
 & $k$ &  0.5  &  0.4820  &  0.1051  & ( 0.3000 ,  0.6233 )
 &  0.1163  &  -0.0180  &  96.7  \\ \cline{2-9}
 & $\phi_3$ &  0.1  &  0.1051  &  0.0266  & ( 0.0700 ,  0.1600 )
 &  0.0277  &  0.0051  &  100.0  \\ \cline{2-9}
 & $u$ &  7.7263  &  7.6486  &  1.7412  & ( 4.7243 ,  10.6616 )
 &  1.6622  &  -0.0777  &  100.0  \\ \cline{2-9}
 & $\xi$ &  -0.2  &  -0.2394  &  0.1787  & ( -0.9372 ,  -0.2193 )
 &  0.2267  &  -0.0394  &  93.3  \\ \cline{2-9}
 & $\sigma$ &  5  &  5.2326  &  1.8908  & ( 3.4637 ,  10.5298 )
 &  1.7658  &  0.2326  &  95.0  \\ \hline

\multirow{9}{*}{ 750 }
 & $\phi_1$ &  0.3  &  0.2993  &  0.0170  & ( 0.2627 ,  0.3307 )
 &  0.0162  &  -0.0007  &  91.7  \\ \cline{2-9}
 & $\phi_2$ &  0.2  &  0.1787  &  0.0448  & ( 0.0573 ,  0.2320 )
 &  0.0523  &  -0.0213  &  96.7  \\ \cline{2-9}
 & $\theta$ &  2  &  2.0449  &  0.2708  & ( 1.6375 ,  2.7055 )
 &  0.2513  &  0.0449  &  96.7  \\ \cline{2-9}
 & $d$ &  1  &  0.9359  &  0.1536  & ( 0.5176 ,  1.1521 )
 &  0.1631  &  -0.0641  &  91.7  \\ \cline{2-9}
 & $k$ &  0.5  &  0.4888  &  0.0925  & ( 0.3000 ,  0.6737 )
 &  0.0984  &  -0.0112  &  96.7  \\ \cline{2-9}
 & $\phi_3$ &  0.1  &  0.1020  &  0.0229  & ( 0.0720 ,  0.1573 )
 &  0.0281  &  0.0020  &  100.0  \\ \cline{2-9}
 & $u$ &  7.7263  &  7.8249  &  1.6102  & ( 4.6059 ,  10.1338 )
 &  1.6988  &  0.0986  &  100.0  \\ \cline{2-9}
 & $\xi$ &  -0.2  &  -0.2222  &  0.1383  & ( -0.6626 ,  -0.1106 )
 &  0.1758  &  -0.0222  &  93.3  \\ \cline{2-9}
 & $\sigma$ &  5  &  5.0831  &  1.5414  & ( 3.3962 ,  8.8558 )
 &  1.4052  &  0.0831  &  96.7  \\ \hline

\multirow{9}{*}{ 1000 }
 & $\phi_1$ &  0.3  &  0.3000  &  0.0144  & ( 0.2420 ,  0.2970 )
 &  0.0145  &  -0.0000  &  91.7  \\ \cline{2-9}
 & $\phi_2$ &  0.2  &  0.1760  &  0.0409  & ( 0.0630 ,  0.2230 )
 &  0.0514  &  -0.0240  &  95.0  \\ \cline{2-9}
 & $\theta$ &  2  &  2.0399  &  0.2316  & ( 1.6985 ,  2.6276 )
 &  0.2139  &  0.0399  &  98.3  \\ \cline{2-9}
 & $d$ &  1  &  0.9300  &  0.1281  & ( 0.5300 ,  1.0387 )
 &  0.1584  &  -0.0700  &  95.0  \\ \cline{2-9}
 & $k$ &  0.5  &  0.4926  &  0.0824  & ( 0.3457 ,  0.6841 )
 &  0.0882  &  -0.0074  &  98.3  \\ \cline{2-9}
 & $\phi_3$ &  0.1  &  0.1036  &  0.0236  & ( 0.0730 ,  0.1600 )
 &  0.0285  &  0.0036  &  98.3  \\ \cline{2-9}
 & $u$ &  7.7263  &  7.7394  &  1.0805  & ( 4.8343 ,  9.2416 )
 &  1.6938  &  0.0131  &  100.0  \\ \cline{2-9}
 & $\xi$ &  -0.2  &  -0.2177  &  0.1087  & ( -0.3421 ,  0.0947 )
 &  0.1372  &  -0.0177  &  93.3  \\ \cline{2-9}
 & $\sigma$ &  5  &  5.0765  &  0.8238  & ( 3.4016 ,  6.6961 )
 &  1.1792  &  0.0765  &  98.3  \\ \hline

\multirow{9}{*}{ 1200 }
 & $\phi_1$ &  0.3  &  0.2999  &  0.0137  & ( 0.2850 ,  0.3392 )
 &  0.0135  &  -0.0001  &  95.0  \\ \cline{2-9}
 & $\phi_2$ &  0.2  &  0.1764  &  0.0478  & ( 0.0692 ,  0.2408 )
 &  0.0511  &  -0.0236  &  91.7  \\ \cline{2-9}
 & $\theta$ &  2  &  2.0247  &  0.1860  & ( 1.6346 ,  2.3487 )
 &  0.1839  &  0.0247  &  95.0  \\ \cline{2-9}
 & $d$ &  1  &  0.9297  &  0.1530  & ( 0.5699 ,  1.1381 )
 &  0.1570  &  -0.0703  &  93.3  \\ \cline{2-9}
 & $k$ &  0.5  &  0.4986  &  0.0917  & ( 0.3000 ,  0.5987 )
 &  0.0808  &  -0.0014  &  95.0  \\ \cline{2-9}
 & $\phi_3$ &  0.1  &  0.1045  &  0.0145  & ( 0.0800 ,  0.1433 )
 &  0.0264  &  0.0045  &  96.7  \\ \cline{2-9}
 & $u$ &  7.7263  &  7.6629  &  0.6811  & ( 4.5641 ,  7.6528 )
 &  1.5515  &  -0.0634  &  98.3  \\ \cline{2-9}
 & $\xi$ &  -0.2  &  -0.2113  &  0.0811  & ( -0.5724 ,  -0.2468 )
 &  0.1190  &  -0.0113  &  91.7  \\ \cline{2-9}
 & $\sigma$ &  5  &  5.0505  &  0.8362  & ( 5.2703 ,  8.7618 )
 &  1.0499  &  0.0505  &  96.7  \\ \hline
\end{tabular}
\label{Table-EEVIMM1}
\end{table*}
\end{landscape}

\begin{landscape}
\begin{table*}[htbp]
\centering
\scriptsize
\setlength{\tabcolsep}{3pt}
\renewcommand{\arraystretch}{0.9}
  \caption{Monte Carlo simulation results ($N=2000$) for the exponential-tailed scenario with sample sizes $n=300,400,500,750,1000,$ and $1200$, reporting Sample mean, BSE, 95\% BCI, RMSE, Bias, and CP.}
\begin{tabular}{|p{2.5cm}|p{1.8cm}|p{1.8cm}|p{2.3cm}|p{1.8cm}|p{2.6cm}|p{2.3cm}|p{2.3cm}|p{1.5cm}|}
\hline
\textbf{n} & \textbf{Parameters} & \textbf{True Value} &
\textbf{Sample Mean} & \textbf{BSE} & \textbf{BCI} &
\textbf{RMSE} & \textbf{Bias} & \textbf{CP} \\ \hline
\multirow{9}{*}{ 300 }
 & $\phi_1$ &  0.3  &  0.2991  &  0.0263  & ( 0.2400 ,  0.3433 )
 &  0.0265  &  -0.0009  &  88.3  \\ \cline{2-9}
 & $\phi_2$ &  0.2  &  0.1838  &  0.0528  & ( 0.0600 ,  0.2533 )
 &  0.0555  &  -0.0162  &  100.0  \\ \cline{2-9}
 & $\theta$ &  2  &  2.1041  &  0.4900  & ( 1.6499 ,  3.5052 )
 &  0.4285  &  0.1041  &  95.0  \\ \cline{2-9}
 & $d$ &  1  &  0.9486  &  0.1575  & ( 0.5276 ,  1.1471 )
 &  0.1851  &  -0.0514  &  100.0  \\ \cline{2-9}
 & $k$ &  0.5  &  0.4806  &  0.1338  & ( 0.3000 ,  0.7092 )
 &  0.1451  &  -0.0194  &  100.0  \\ \cline{2-9}
 & $\phi_3$ &  0.1  &  0.1231  &  0.0199  & ( 0.1000 ,  0.1600 )
 &  0.0287  &  0.0231  &  100.0  \\ \cline{2-9}
 & $u$ &  7.7263  &  6.5903  &  1.5233  & ( 4.3739 ,  9.8278 )
 &  1.6983  &  -1.1360  &  90.0  \\ \cline{2-9}
 & $\xi$ &  0  &  -0.0840  &  0.2398  & ( -0.7003 ,  0.2487 )
 &  0.2679  &  -0.0840  &  90.0  \\ \cline{2-9}
 & $\sigma$ &  5  &  5.5982  &  2.5081  & ( 3.7240 ,  13.4792 )
 &  1.9825  &  0.5982  &  95.0  \\ \hline

\multirow{9}{*}{ 400 }
 & $\phi_1$ &  0.3  &  0.2998  &  0.0229  & ( 0.2575 ,  0.3500 )
 &  0.0227  &  -0.0002  &  90.0  \\ \cline{2-9}
 & $\phi_2$ &  0.2  &  0.1815  &  0.0476  & ( 0.0575 ,  0.2400 )
 &  0.0543  &  -0.0185  &  100.0  \\ \cline{2-9}
 & $\theta$ &  2  &  2.0883  &  0.4305  & ( 1.7769 ,  3.4190 )
 &  0.3672  &  0.0883  &  96.7  \\ \cline{2-9}
 & $d$ &  1  &  0.9418  &  0.1498  & ( 0.5850 ,  1.1931 )
 &  0.1759  &  -0.0582  &  95.0  \\ \cline{2-9}
 & $k$ &  0.5  &  0.4863  &  0.0950  & ( 0.3000 ,  0.6942 )
 &  0.1272  &  -0.0137  &  96.7  \\ \cline{2-9}
 & $\phi_3$ &  0.1  &  0.1113  &  0.0214  & ( 0.0775 ,  0.1600 )
 &  0.0259  &  0.0113  &  100.0  \\ \cline{2-9}
 & $u$ &  7.7263  &  7.2613  &  1.3318  & ( 5.1089 ,  10.9700 )
 &  1.5793  &  -0.4650  &  100.0  \\ \cline{2-9}
 & $\xi$ &  0  &  -0.0351  &  0.2458  & ( -0.3823 ,  0.6159 )
 &  0.2545  &  -0.0351  &  88.3  \\ \cline{2-9}
 & $\sigma$ &  5  &  5.3091  &  1.5685  & ( 2.2130 ,  8.7377 )
 &  1.9389  &  0.3091  &  95.0  \\ \hline

\multirow{9}{*}{ 500 }
 & $\phi_1$ &  0.3  &  0.3001  &  0.0211  & ( 0.2900 ,  0.3720 )
 &  0.0204  &  0.0001  &  96.7  \\ \cline{2-9}
 & $\phi_2$ &  0.2  &  0.1815  &  0.0497  & ( 0.0680 ,  0.2420 )
 &  0.0527  &  -0.0185  &  98.3  \\ \cline{2-9}
 & $\theta$ &  2  &  2.0730  &  0.3438  & ( 1.6579 ,  2.8782 )
 &  0.3195  &  0.0730  &  96.7  \\ \cline{2-9}
 & $d$ &  1  &  0.9442  &  0.1510  & ( 0.5795 ,  1.1403 )
 &  0.1665  &  -0.0558  &  98.3  \\ \cline{2-9}
 & $k$ &  0.5  &  0.4869  &  0.1214  & ( 0.3000 ,  0.7560 )
 &  0.1127  &  -0.0131  &  100.0  \\ \cline{2-9}
 & $\phi_3$ &  0.1  &  0.1038  &  0.0249  & ( 0.0720 ,  0.1600 )
 &  0.0276  &  0.0038  &  100.0  \\ \cline{2-9}
 & $u$ &  7.7263  &  7.7729  &  1.3695  & ( 4.6213 ,  10.3451 )
 &  1.7118  &  0.0466  &  100.0  \\ \cline{2-9}
 & $\xi$ &  0  &  -0.0272  &  0.2007  & ( -0.3553 ,  0.4491 )
 &  0.2430  &  -0.0272  &  91.7  \\ \cline{2-9}
 & $\sigma$ &  5  &  5.1276  &  1.4109  & ( 2.9010 ,  8.6027 )
 &  1.8157  &  0.1276  &  96.7  \\ \hline

\multirow{9}{*}{ 750 }
 & $\phi_1$ &  0.3  &  0.2996  &  0.0167  & ( 0.2653 ,  0.3320 )
 &  0.0170  &  -0.0004  &  93.3  \\ \cline{2-9}
 & $\phi_2$ &  0.2  &  0.1811  &  0.0352  & ( 0.0520 ,  0.1867 )
 &  0.0498  &  -0.0189  &  93.3  \\ \cline{2-9}
 & $\theta$ &  2  &  2.0341  &  0.3459  & ( 1.9701 ,  3.3422 )
 &  0.2432  &  0.0341  &  98.3  \\ \cline{2-9}
 & $d$ &  1  &  0.9447  &  0.1072  & ( 0.5376 ,  0.9723 )
 &  0.1548  &  -0.0553  &  93.3  \\ \cline{2-9}
 & $k$ &  0.5  &  0.4899  &  0.0882  & ( 0.3078 ,  0.6639 )
 &  0.0987  &  -0.0101  &  100.0  \\ \cline{2-9}
 & $\phi_3$ &  0.1  &  0.1011  &  0.0248  & ( 0.0747 ,  0.1600 )
 &  0.0290  &  0.0011  &  100.0  \\ \cline{2-9}
 & $u$ &  7.7263  &  7.9020  &  1.4658  & ( 6.0089 ,  11.6964 )
 &  1.7988  &  0.1757  &  100.0  \\ \cline{2-9}
 & $\xi$ &  0  &  -0.0235  &  0.1346  & ( -0.3291 ,  0.2204 )
 &  0.2017  &  -0.0235  &  96.7  \\ \cline{2-9}
 & $\sigma$ &  5  &  5.1282  &  1.1393  & ( 3.9547 ,  8.8505 )
 &  1.5137  &  0.1282  &  100.0  \\ \hline

\multirow{9}{*}{ 1000 }
 & $\phi_1$ &  0.3  &  0.2997  &  0.0145  & ( 0.2830 ,  0.3400 )
 &  0.0144  &  -0.0003  &  90.0  \\ \cline{2-9}
 & $\phi_2$ &  0.2  &  0.1783  &  0.0463  & ( 0.0620 ,  0.2410 )
 &  0.0503  &  -0.0217  &  91.7  \\ \cline{2-9}
 & $\theta$ &  2  &  2.0261  &  0.2005  & ( 1.5864 ,  2.3770 )
 &  0.2019  &  0.0261  &  98.3  \\ \cline{2-9}
 & $d$ &  1  &  0.9356  &  0.1828  & ( 0.6231 ,  1.3426 )
 &  0.1552  &  -0.0644  &  90.0  \\ \cline{2-9}
 & $k$ &  0.5  &  0.4991  &  0.0734  & ( 0.3000 ,  0.6285 )
 &  0.0850  &  -0.0009  &  96.7  \\ \cline{2-9}
 & $\phi_3$ &  0.1  &  0.1015  &  0.0215  & ( 0.0700 ,  0.1560 )
 &  0.0281  &  0.0015  &  100.0  \\ \cline{2-9}
 & $u$ &  7.7263  &  7.9339  &  1.2931  & ( 4.7811 ,  9.5589 )
 &  1.7606  &  0.2076  &  100.0  \\ \cline{2-9}
 & $\xi$ &  0  &  -0.0046  &  0.1647  & ( -0.1559 ,  0.4905 )
 &  0.1571  &  -0.0046  &  93.3  \\ \cline{2-9}
 & $\sigma$ &  5  &  5.0040  &  0.8749  & ( 2.3807 ,  5.8911 )
 &  1.2046  &  0.0040  &  95.0  \\ \hline

\multirow{9}{*}{ 1200 }
 & $\phi_1$ &  0.3  &  0.2995  &  0.0130  & ( 0.2716 ,  0.3225 )
 &  0.0132  &  -0.0005  &  95.0  \\ \cline{2-9}
 & $\phi_2$ &  0.2  &  0.1791  &  0.0456  & ( 0.0683 ,  0.2317 )
 &  0.0489  &  -0.0209  &  95.0  \\ \cline{2-9}
 & $\theta$ &  2  &  2.0202  &  0.1815  & ( 1.8010 ,  2.5283 )
 &  0.1870  &  0.0202  &  98.3  \\ \cline{2-9}
 & $d$ &  1  &  0.9378  &  0.1310  & ( 0.5745 ,  1.0487 )
 &  0.1492  &  -0.0622  &  91.7  \\ \cline{2-9}
 & $k$ &  0.5  &  0.5046  &  0.0763  & ( 0.3947 ,  0.6988 )
 &  0.0804  &  0.0046  &  100.0  \\ \cline{2-9}
 & $\phi_3$ &  0.1  &  0.1040  &  0.0185  & ( 0.0816 ,  0.1600 )
 &  0.0267  &  0.0040  &  100.0  \\ \cline{2-9}
 & $u$ &  7.7263  &  7.7340  &  0.7002  & ( 5.2708 ,  8.3214 )
 &  1.5637  &  0.0077  &  100.0  \\ \cline{2-9}
 & $\xi$ &  0  &  -0.0100  &  0.1075  & ( -0.0164 ,  0.3899 )
 &  0.1344  &  -0.0100  &  96.7  \\ \cline{2-9}
 & $\sigma$ &  5  &  5.0189  &  0.6030  & ( 2.9617 ,  5.2616 )
 &  1.0363  &  0.0189  &  100.0  \\ \hline
\end{tabular}
\label{Table-EEVIMM2}
\end{table*}
\end{landscape}

\end{document}